\documentclass{ws-m3as}

\usepackage{babel}
\usepackage{float}
\usepackage{subfig}
\usepackage[T1]{fontenc}
\usepackage{color}
\usepackage{mathrsfs}
\usepackage{mathtools}
\usepackage{amsmath}
\usepackage{amssymb}
\usepackage{graphicx}
\usepackage[linesnumbered,ruled,vlined]{algorithm2e}
\usepackage[nolists,nofigures,notables,tablesfirst]{endfloat}
\usepackage[super, compress, comma]{natbib}
\usepackage{chngcntr}
\usepackage{ragged2e}

\counterwithout{figure}{section}
\begin{document}

\markboth{Paolo Bartesaghi and Ernesto Estrada}{Where to cut to delay a pandemic with minimum disruption?}

%
%

\title{Where to cut to delay a pandemic with minimum disruption? Mathematical analysis based on the SIS model}

\author{Paolo Bartesaghi}

\address{Department of Statistics and Quantitative Methods,\\
	University of Milano - Bicocca, Via Bicocca degli Arcimboldi 8, 20126, Milano, Italy.\\
	paolo.bartesaghi@unimib.it}

\author{Ernesto Estrada\footnote{Corresponding author: Ernesto Estrada, email: estrada@ifisc.uib-csic.es}}

\address{Institute of Mathematics and Applications, University of Zaragoza,\\
	     Pedro Cerbuna 12, Zaragoza 50009, Spain;\\
	     ARAID Foundation,\\ Government
	     of Aragón, Spain; Institute for Cross-Disciplinary Physics and Complex
	     Systems (IFISC, UIB-CSIC),\\
	     Campus Universitat de les Illes Balears E-07122, Palma de Mallorca, Spain.\\
	     estrada@ifisc.uib-csic.es}

\maketitle

\begin{history}
\accepted{July 16, 2021}
\end{history}

\begin{abstract}
We consider the problem of modifying a network topology in such a way as to delay the
propagation of a disease with minimal disruption of the network capacity
to reroute goods/items/passengers. We find an approximate
solution to the Susceptible-Infected-Susceptible (SIS) model, which
constitutes a tight upper bound to its exact solution. This upper
bound allows direct structure-epidemic dynamic relations via the total
communicability function. Using this approach we propose a strategy
to remove edges in a network that significantly delays the propagation
of a disease across the network with minimal disruption of its capacity
to deliver goods/items/passengers. We apply this strategy to the analysis
of the U.K. airport transportation network weighted by the number
of passengers transported in the year 2003. We find that the removal
of all flights connecting four origin-destination pairs in the U.K.
delays the propagation of a disease by more than 300\%, with a minimal
deterioration of the transportation capacity of this network. These
time delays in the propagation of a disease represent an important
non-pharmaceutical intervention to confront an epidemics, allowing
for better preparations of the health systems, while keeping the economy
moving with minimal disruptions. 
\end{abstract}



\section{Introduction}

Epidemics propagate through networks \cite{networks_epidemics}. They
include international transportation webs \cite{epidemics_airports},
nationwide and urban commuting systems \cite{epidemics_commuting},
as well as face-to-face networks of human contacts \cite{epidemics_face_to_face}.
Temporarily disrupting these networks is the first election to avoid
the propagation of an epidemic from local to global scales \cite{edge_removal_3,edge_removal_2,edge_removal_1}.
The problem is that these networks also move our economy and society.
Arguably a network exists to transport ``something'' among its nodes.
Therefore, the disruption of transportation networks affects the flow
of raw materials needed for production, of goods needed for consumption
and of people who directly or indirectly participate in the economic
life of modern society. During the COVID-19 pandemic \cite{COVID-19_3,COVID-19_2},
expanding from Wuhan in China to the rest of the globe since February
2020, both human and economic damages have been catastrophic across
the world \cite{COVID-19_damages_4,COVID-19_damages_1,COVID-19_damages_2,COVID-19_damages_3}.

One lesson learned from this pandemic is that sometimes delaying the
propagation of the infection a few days allows for a better preparation
of the health services which impact significantly in saving lives
\cite{delaying_3,delaying_1,delaying_2}. In order to delay such propagation
we have to act directly over the networks which the disease uses to
expand and, in this sense, mathematical modeling has played an important
role \cite{modeling_1,modeling_4,modeling_2,modeling_3}. In practice,
we can cut some of the connections between the different nodes of
these networks with the hope of delaying the pandemic. Typically,
we are talking about canceling or reducing international and national
flights, isolating regions of a country and/or neighborhoods of a
city, and/or limiting the sizes of social groups allowed \cite{delaying_3,delaying_1,delaying_2}.
The question is then: ``Where to cut?'', thinking simultaneously
in delaying the pandemic and not affecting dramatically the flow of
goods/items/passengers through the network.

Let us illustrate this situation with a toy example.
In the next section we will motivate the use of the "Suceptible-Infected-Susceptible" (SIS) model in the context of diseases propagating on an airport network. Then, let us consider a disease propagating across an airport network like the one illustrated in Fig. \ref{Toy network}. In this scenario it would be tempting to cut
the edge between the nodes 1 and 12 to delay the propagation. Indeed,
the time at which the whole network is infected in a SIS \cite{kiss2017mathematics,Bullo,SzaboSimon2014} scenario is delayed by 9.45\%
with respect to the original network, but you have increased the average
shortest path distance by 36.3\%, making the network much inefficient.
If instead you decided to cut 5-6, you increase the SIS time of global
infection by 8.92\% but also the average shortest path by 25.2\%.
In contrast, removing the edge (2,3) increases the SIS time of global
infection by 3.15\% with an increase of only 0.59\% in the average
shortest path.

\begin{figure}[H]
	\begin{centering}
		\includegraphics[width=0.55\textwidth]{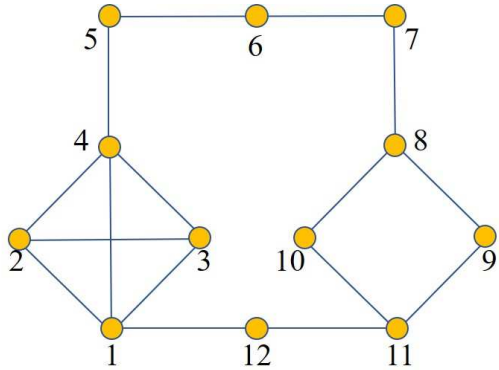}
		\par\end{centering}
	\caption{A toy network illustrating two communities (squares) connected by
		two paths of 2 and 4 edges, respectively. The toy model is used to
		illustrate the strategies of delaying disease propagation with minimum
		connectivity cost by cutting strategic edges in graphs.}
	
	\label{Toy network}
\end{figure}

The concepts stated before are clarified in the next sections of this
paper, but they are used here to exemplify the lack of triviality
of the problem in question. After the definition of all concepts
and notation used in this work, we state and prove the main result,
namely an approximate SIS model whose solution represents
an upper bound to the exact solution of this model, which is always below
the diverging solution of the linearized model. One of the main advantages
of this upper bound, apart from representing a worse case scenario
for the propagation of a SIS disease, is that it directly connects
the structure of the network with the disease dynamics. That is, we
show here that the upper bound found here for the SIS model can be
expressed in terms of an exponential function of the adjacency matrix,
which is known to capture the contributions of subgraphs of a network
to a property, e.g., node infectivity, via the use of walks in graphs.
These functions, known nowadays as communicability functions, have
found many applications across the disciplines. Using our structural-transparent
upper bound to the SIS model we study the propagation of a disease
across the network of commercial airports in the U.K. We consider
a network of airports with edges weighed by the number of passengers
transported in year 2003. We devised here a strategy to remove edges
in a network which delays the disease propagation with minimum disruption
of the capacity of the network to reroute goods/items/passengers.
For the U.K. airport network we found that by removing 4 edges, i.e,
removing all flights connecting 4 pairs of origin-destination places,
a disease can be delayed by more than 300\% relative to the original
network without disruption of the network efficiency to diffuse goods/items/passengers
or to reroute them by shortest paths connections in the network.

\section{Motivations - SIS model and airport networks}

In this work we deal with the use of the SIS model, which is mathematically
described in the next section. The reasons why we consider SIS instead
of, for instance, the Susceptible-Infected-Recovered (SIR) model are
presented in this section. In the Introduction we have mentioned our
interest in modeling and understanding situations in which a pandemic,
like the current SARS-CoV-2, is affecting the global population
or a significant part of it. Like in most of complex systems we can
model this situation from a multi-scale perspective. Let us simplify this setting and consider for instance the three scales illustrated in Fig. \ref{Scales}

\begin{figure}[H]
	\begin{centering}
		\includegraphics[width=0.90\textwidth]{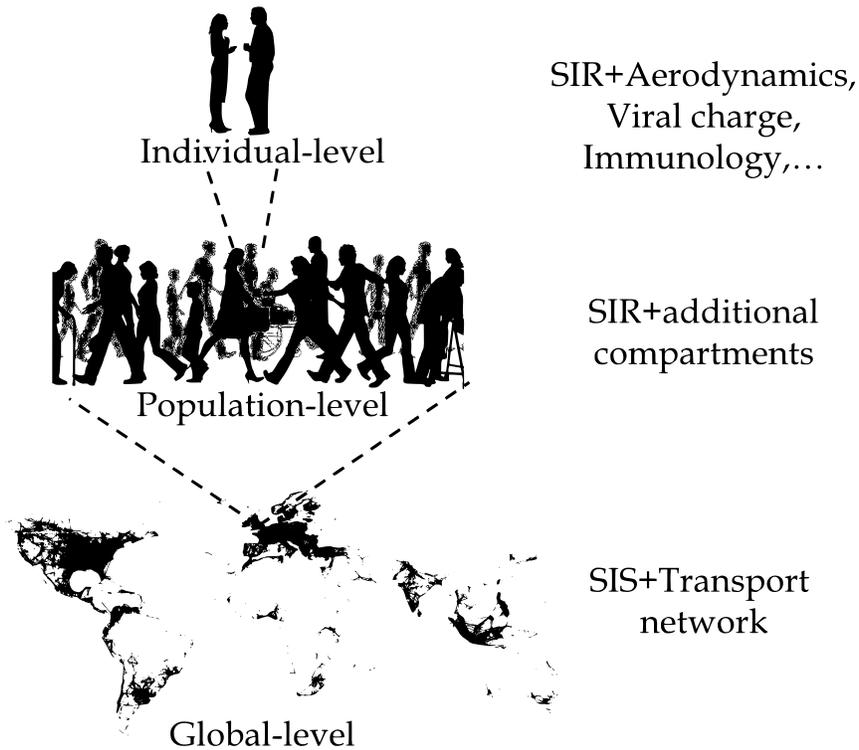}
	\par\end{centering}
	\caption{Schematic illustration of three different scales at which the propagation of a disease can be modeled in a complex system.}
	\label{Scales}
\end{figure}

At the smallest of the three we have the physical contact between two individuals. At the intermediate one we have the potential contacts between a significant part of the population. At the largest scale we have the interconnection between regions of the world, e.g., airports, cities, countries, etc.

To model the transmission of a virus at the smallest of the three scales considered we need not only a simple epidemiological model, which considers whether one individual is infected and the other is susceptible or have being recovered from the disease. We also need to consider particle aerodynamics, the viral charge of the infected individual, the immunological situation of the susceptible one, etc. In the case of transmission of diseases like flu, SARS, or SARS-CoV-2, between individuals it is well-known that the most appropriate model is the SIR one. Pairwise approximations for SIR-type network epidemics are analysed, for instance, in Keeling and Eames \cite{networks_epidemics}, in House and Keeling \cite{HouseKeeling2011} and in R\"{o}st et al. \cite{Vizi2018}. An extension of the edge-based compartmental model to SIR epidemics with general but independent distributions for time to transmission and duration of the infectious period is proposed in Sherborne et al. \cite{Miller2017}.

In modeling the intermediate situation, in case of viral diseases,
it is customary to use modified SIR models, where the population is
split into several compartments which include the classes of susceptible,
infected and recovered. The reader is referred to the review \cite{modeling_4}
and references therein for details in the case of SARS-CoV-2.

Now, in the case of the largest scale, we have some particularities which
need to be considered. In this largest scale scenario the system is represented by a network
in which the nodes are the countries, cities or airports. Let us consider
the case of airports. Then, an airport is susceptible to the disease
if none of the passengers in that airport at a given time is infected.
This airport can become infected due to the fact that infected passenger(s)
come from other nearest neighbor airports (see Fig. \ref{SI_transform}(a)).
It is clear that an airport cannot be considered ``recovered'' in
the sense of creating immunity, at least in the absence of quarantines
in this airport (which are not considered here). Therefore, an infected
airport can become susceptible again if the infected passenger(s)
that were located in that airport move away from it (see Fig. \ref{SI_transform}(b)).

\begin{figure}[H]
	\begin{centering}
		\subfloat[]{\begin{centering}
				\includegraphics[width=0.5\textwidth]{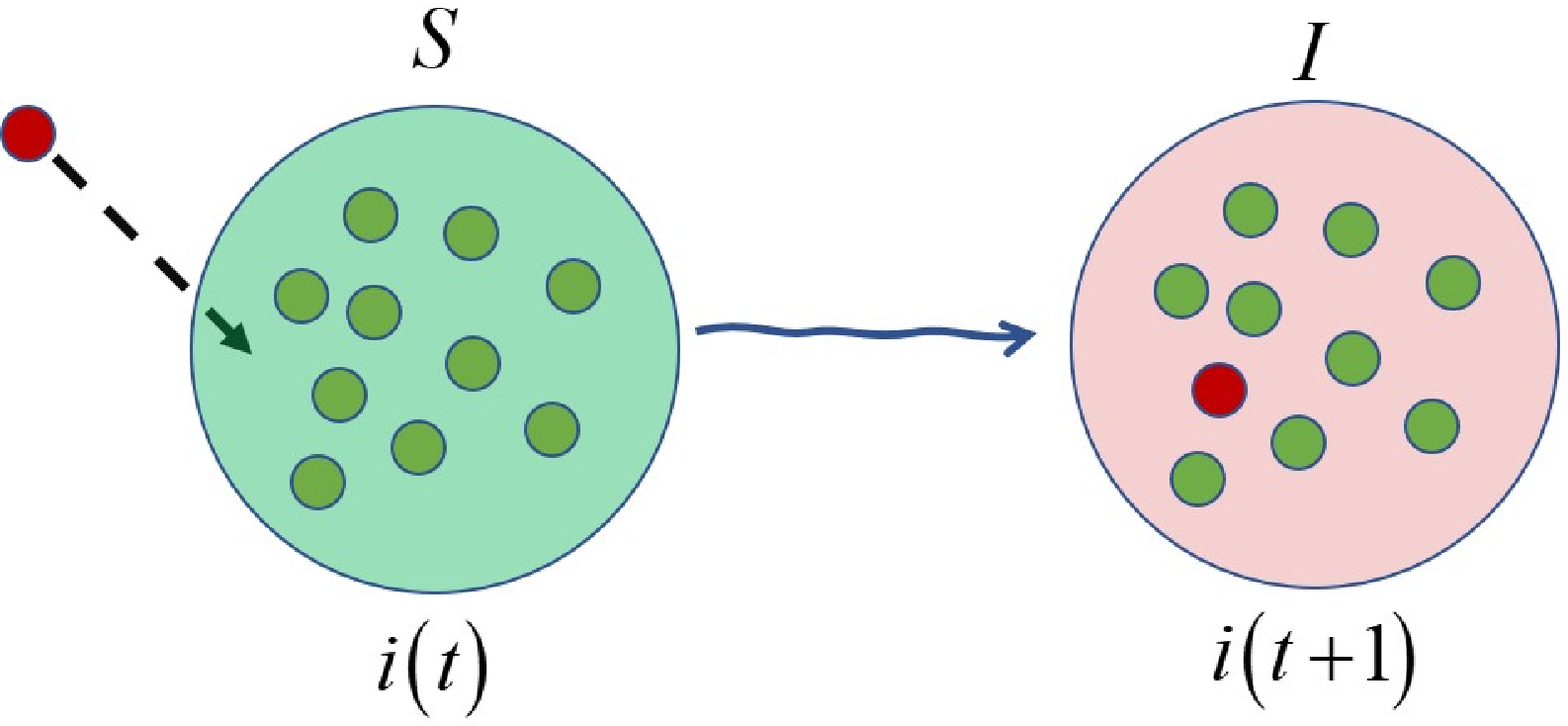}
				\par\end{centering}
		}
		\par\end{centering}
	\begin{centering}
		\subfloat[]{\begin{centering}
				\includegraphics[width=0.5\textwidth]{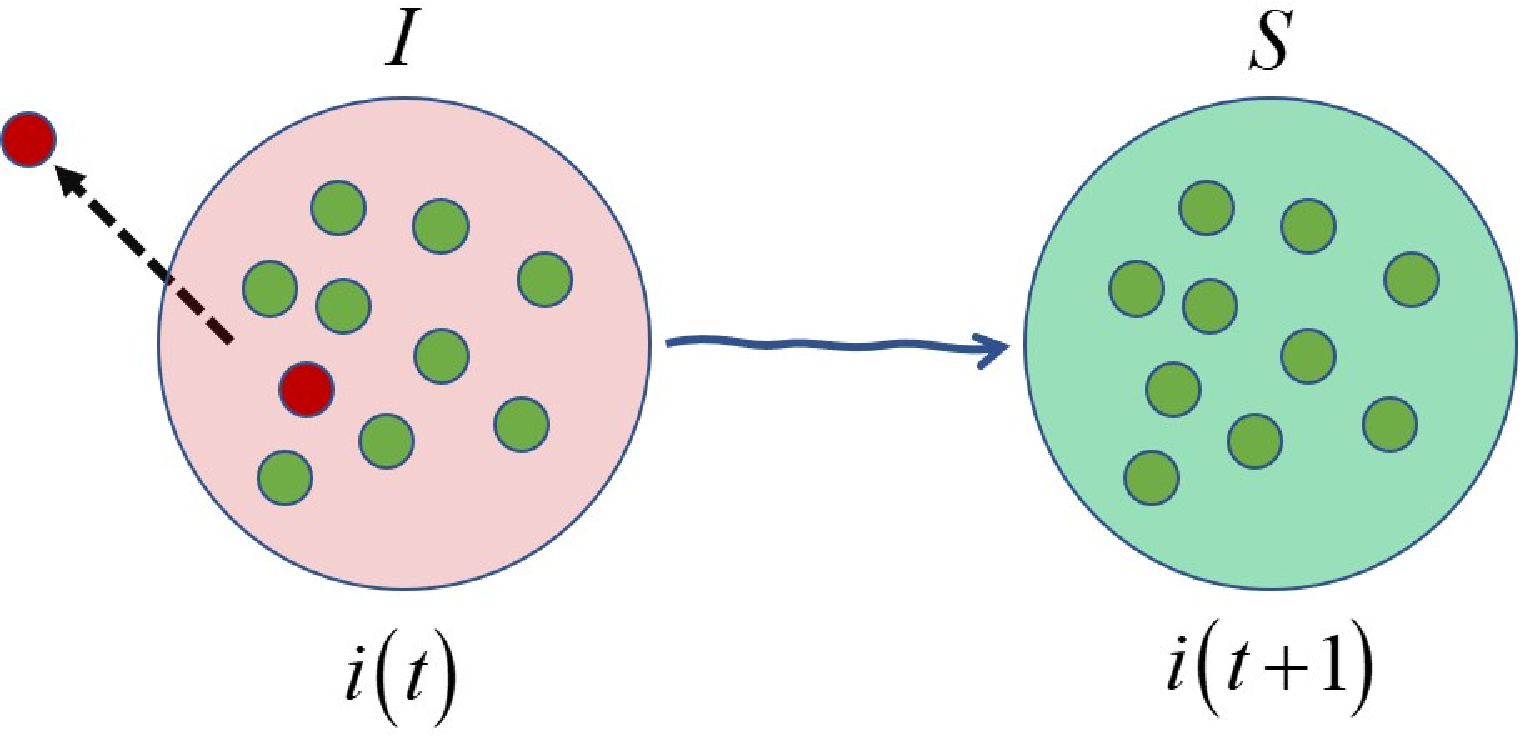}
				\par\end{centering}
		}
		\par\end{centering}
	\caption{Illustration of the $S\rightarrow I$ (a) and $I\rightarrow S$ (b)
		transformations of an airport labeled as $i$ in the airport network.}
	\label{SI_transform}
\end{figure}

This SIS strategy has been applied to similar situations for instance by Omi\'{c} and Van Mieghem \cite{VanMieghem2010} who considered a city as a node which
can have infection introduced over the air transportation network
from other cities. In this scenario they considered the SIS model
as the adequate modeling tool. In another work Matamalas et al. \cite{Matamalas2018}
used again SIS for the worldwide air transportation network, with
the goal of identifying the most important connections between airports
for the spreading of epidemics and evaluate the epidemic incidence
after its deactivation. Sanders et al. \cite{SIS_airports_1} analyzed the
spread of an infection disease through 23 subpopulations via (documented)
air traffic data, and considering that the country is internationally
quarantined. Similarly, Qu and Wang \cite{SIS_airports_5}, Scaman et al.
\cite{Scaman2014}, Onoue et al. \cite{Onoue2020}, and Ruan et al., \cite{SIS_airports_2}, Meloni et al. \cite{Meloni2009} and Ye et al. \cite{Ye2020}, among others, used SIS for modeling diseases propagating through airports in a network system.

\section{Preliminaries}

We consider here (weighted) graphs $\Gamma=(V,E,W,\varphi)$ where
$V$ is the set of vertices (nodes), $E$ is the set of edges and
$W$ is a set of weights $w_{ij}\in W$, such that $w_{ij}\in\mathbb{R}^{+}$
\cite{Estrada2012book,Newman2010}. The weights are assigned to the
edges by the surjective mapping $\varphi:E\rightarrow W$. In the
case of unweighted graphs we consider that $w_{ij}=1$ for $\left(i,j\right)\in E$.
We always consider $\#V=n$ and $\#E=m.$ All graphs here are undirected,
therefore their adjacency matrices $A$ are symmetric. Let $k_{i}$
denotes the (weighted) degree of the node $i$. Then, ${\bf k}=[k_{1},\dots,k_{n}]^{T}$
is the degree vector. The eigenvalues and eigenvectors of the adjacency
matrix are designated, respectively, by $\lambda_{i},\ i=1,\dots,n$,
$\lambda_{1}>\lambda_{2}\geq\dots\geq\lambda_{n}$ and $\psi_{i},\ i=1,\dots,n$.
To complete the notation we use ${\bf u}=[1,\dots,1]^{T}$ as the
all ones vector, $U={\bf u}{\bf u}^{T}$, ${\bf 0}=[0,\dots,0]^{T}$
and $I$ the identity matrix

In the implementation of compartmental epidemiological models we consider
$x_{i}(t)$ to be the probability that node $i$ is infected at time
$t$ \cite{Bullo}. For the whole network, we define the vector ${\bf x}(t)=[x_{1}(t),\dots,x_{n}(t)]^{T}$.
We designate by $p$ the initial probability of being infected, by
$q=1-p$ the initial probability of being healthy; $\beta$ is the
infection rate \textit{per link,} $\gamma$ the recovering rate, $\beta_{e}=\beta/\gamma$
the effective infectivity rate and $\tau$ the epidemic threshold, that is the critical effective  rate above which the disease infects a non-zero fraction of the whole population.

We consider here a SIS epidemiological model on the graph $\varGamma$ \cite{kiss2017mathematics,Bullo}.
In this case an infected node can infect any of its nearest susceptible
neighbors which then become infected with infection rate $\beta>0$.
The infected node can recover with recovery rate $\gamma>0$
and become susceptible again (see Fig. \ref{SIS_model}).

\begin{figure}[H]
	\begin{centering}
		\includegraphics[width=0.45\textwidth]{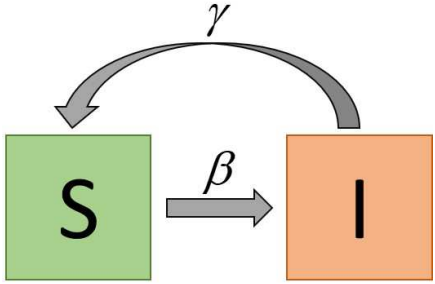}
		\par\end{centering}
	\caption{Diagram illustrating the flux of susceptible nodes (S) that
		become infected/infectious (I) with rate $\beta$ and which can cure
		without immunity and become susceptible again with rate $\gamma$.}
	\label{SIS_model}
\end{figure}

The evolution of the probability of getting infected for a node $i$
is described by \cite{Bullo}:
\begin{equation}
	\dfrac{dx_{i}\left(t\right)}{dt}=\beta\left(1-x_{i}\left(t\right)\right)\sum_{j\in\mathcal{N}_{i}}A_{ij}x_{j}\left(t\right)-\gamma x_{i}\left(t\right),\ t\geq {t_0},\label{eq:SI_original-1}
\end{equation}
where $A_{ij}$ are the entries of the adjacency matrix for the pair
of nodes $i$ and $j$, and $\mathcal{N}_{i}$ is the set of nearest neighbors
of $i$. In matrix-vector form, it becomes:

\begin{equation}
	\dot{\mathbf{x}}\left(t\right)=\dfrac{d\mathbf{x}\left(t\right)}{dt}=\beta\left[I_{N}-\textnormal{diag}\left(\mathbf{x}\left(t\right)\right)\right]A\mathbf{x}\left(t\right)-\gamma\mathbf{x}\left(t\right),
\end{equation}
with initial condition $\mathbf{x}\left(0\right)=\mathbf{x}_{0}=p{\bf u}.$
The following two results can be found in \cite{Bullo} and characterize
the behavior of network SIS below and over the epidemic threshold,
which we adapt here to the case of undirected networks only:
\begin{theorem}
	If $\beta\lambda_{1}/\gamma<1$ we have the following results:
	
	(i) if ${\bf x}_{0}\in[0,1]^{n}$ then ${\bf x}(t)\in[0,1]^{n}$ for
	all $t>0$;
	
	(ii) there exists a unique equilibrium point ${\bf x}^{\star}={\bf 0}$
	and it is exponentially stable;
	
	(iii) the linearization of the model around the point ${\bf 0}$ is
	given by
	
	\begin{equation}
		\dot{\mathbf{x}}\left(t\right)=\left(\beta A-\gamma I\right)\mathbf{x}\left(t\right).
	\end{equation}
\end{theorem}

\begin{theorem}
	If $\beta\lambda_{1}/\gamma>1$ we have the following results:
	
	(i) if ${\bf x}_{0}\in[0,1]^{n}$ then ${\bf x}(t)\in[0,1]^{n}$ for
	all $t>0$ and if ${\bf x}_{0}>{\bf 0}$ then ${\bf x}(t)>{\bf 0}$
	for all $t>0$;
	
	(ii) there exists an equilibrium point ${\bf x}^{\star}={\bf 0}$,
	the epidemic outbreak, exponentially unstable;
	
	(iii) there exists an equilibrium point ${\bf x}^{\star}\neq{\bf 0}$,
	the endemic state, exponentially stable such that:
	
	\begin{equation}
		{\bf x}^{\star}\rightarrow\left\{ \begin{array}{cc}
			a\left(\tfrac{\beta}{\gamma}\lambda_{1}-1\right)\psi_{1} & \quad \quad \textnormal{\ensuremath{\mathnormal{\textnormal{ if }}}}\beta\rightarrow\left(\tfrac{\gamma}{\lambda_{1}}\right)^{+}\\
			{\bf u}-\frac{\gamma}{\beta}\left(\textnormal{diag}\left({\bf k}\right)\right)^{-1} & \textnormal{\ensuremath{\mathnormal{\textnormal{ if }}}}\gamma\rightarrow0,
		\end{array}\right.
	\end{equation}
	
	where $a=\tfrac{\left\Vert \psi_{1}\right\Vert ^{2}}{\psi_{1}^{T}\textnormal{diag}\left(\psi_{1}\right)\psi_{1}}$.
\end{theorem}

\subsection{Lee-Tenneti-Eun approximation of the SI model}

Using similar notations as before, the Susceptible-Infected (SI) model
is written as\cite{Bullo}:

\begin{equation}
	\dfrac{dx_{i}\left(t\right)}{dt}=\beta\left(1-x_{i}\left(t\right)\right)\sum_{j\in\mathcal{N}_{i}}A_{ij}x_{j}\left(t\right), t\geq {t_0},\label{eq:SI_original}
\end{equation}
which in matrix-vector form becomes:

\begin{equation}
	\dfrac{d{\bf x}\left(t\right)}{dt}=\beta\left[I_{N}-\textnormal{diag}\left({\bf x}\left(t\right)\right)\right]A{\bf x}\left(t\right),
\end{equation}
with initial condition ${\bf x}\left(0\right)={\bf x}_{0}$. It is well-known
that the linearization of the model around the point {\bf 0} is given by

\begin{equation}
	\dfrac{d{\bf x}\left(t\right)}{dt}=\beta A\,{\bf x}\left(t\right)\label{eq:model}
\end{equation}
which is exponentially unstable. Lee-Tenneti-Eun (LTE) \cite{Lee_et_al}
rewrote the SI equation as

\begin{equation}
	\dfrac{1}{1-x_{i}(t)}\dfrac{dx_{i}\left(t\right)}{dt}=\beta\sum_{j\in\mathcal{N}_{i}}A_{ij}\left(1-e^{-\left(-\log\left(1-x_{j}\left(t\right)\right)\right)}\right),
\end{equation}
which is equivalent to

\begin{equation}
	\dfrac{dy_{i}\left(t\right)}{dt}=\beta\sum_{j\in\mathcal{N}_{i}}A_{ij}f\left(y_{j}\left(t\right)\right),
\end{equation}
where $y_{i}\left(t\right)\coloneqq g\left(x_{i}\left(t\right)\right)=-\log\left(1-x_{i}\left(t\right)\right)\in\left[0,\infty\right]$,
$f\left(y\right)\coloneqq1-e^{-y}=g^{-1}\left(y\right)$.

They then considered the following linearized version of the previous nonlinear equation

\begin{equation}
	\dfrac{d\hat{\mathbf{y}}\left(t\right)}{dt}=\beta A\textnormal{diag}\left(1-\mathbf{x}\left(t_{0}\right)\right)\hat{\mathbf{y}}\left(t\right)+\beta\mathbf{h}\left(\mathbf{x}\left(t_{0}\right)\right),
\end{equation}
where $\hat{{\bf x}}\left(t\right)=f\left(\hat{{\bf y}}\left(t\right)\right)$
in which $\hat{{\bf x}}\left(t\right)$ is the approximate solution to the
SI model, $\hat{\mathbf{y}}\left(t_{0}\right)=g\left(\mathbf{x}\left(t_{0}\right)\right)$
and $\mathbf{h}\left(\mathbf{x}\right)\coloneqq\mathbf{x}+\left(1-\mathbf{x}\right)\log\left(\mathbf{u}-\mathbf{x}\right).$

They then proved the following result. \cite{Lee_et_al}
\begin{theorem}
	For any $t\geq t_{0}$,
	
	\begin{equation}
		\mathbf{x}\left(t\right)\preceq\mathbf{\hat{x}}\left(t\right)=f\left(\hat{\mathbf{y}}\left(t\right)\right)\preceq\mathbf{\tilde{x}}\left(t\right),
	\end{equation}
where
	
	\begin{equation}
		\begin{split}\hat{\mathbf{y}}\left(t\right) & =e^{\beta\left(t-t_{0}\right)A\textnormal{diag}\left(1-\mathbf{x}\left(t_{0}\right)\right)}g\left(\mathbf{x}\left(t_{0}\right)\right)\\
			& +\sum_{k=0}^{\infty}\dfrac{\left(\beta\left(t-t_{0}\right)\right)^{k+1}}{\left(k+1\right)!}\left[A\textnormal{diag}\left(1-\mathbf{x}\left(t_{0}\right)\right)\right]^{k}A\mathbf{h}\left(\mathbf{x}\left(t_{0}\right)\right),
		\end{split}
	\end{equation}
is the solution of the approximate model of LTE and
	
	\begin{equation}
		\mathbf{\tilde{x}}\left(t\right)=e^{\beta\left(t-t_{0}\right)A}\mathbf{x}\left(t_{0}\right),
	\end{equation}
is the solution of the linearized model.
\end{theorem}

When $t_{0}=0$ and $x_{i}\left(0\right)=p$, $\forall i=1,2,\ldots,N$ the previous equation is transformed to

\begin{equation}
	\hat{\mathbf{y}}\left(t\right)=\left(1/q-1\right)e^{q\beta tA}\mathbf{u}-\left(1/q-1+\log q\right)\mathbf{u}. \label{LTE solution}
\end{equation}

\section{Mathematical results}

\subsection{Tight upper bound for the SIS model}

We start here by rewriting the SIS model in Eq. (\ref{eq:SI_original-1}) with the use of the variables $y_{i}\left(t\right)$, such that $x_{i}\left(t\right)=1-e^{-y_{i}\left(t\right)}$
and $\dot{x}_{i}(t)=e^{-y_{i}\left(t\right)}\dot{y}_{i}\left(t\right)$.
Then, we have

\begin{equation}
	\dot{y}_{i}\left(t\right)=\left[\beta\sum_{j=1}^{n}A_{ij}\left(1-e^{-y_{j}\left(t\right)}\right)\right]-\gamma\left(e^{y_{i}}-1\right),
\end{equation}
which can be transformed to

\begin{equation}
	\dot{y}_{i}\left(t\right)=\sum_{j=1}^{n}\left[ \beta A_{ij}\left(1-e^{-y_{j}\left(t\right)}\right)-\gamma \delta_{ij} \left(e^{y_{j}}-1\right)\right],\label{eq:SIS_new}
\end{equation}
using Kronecker $\delta_{ij}$ function.

Let us remark that $f\left(y\right)=1-e^{-y}$ is an increasing concave
function. Then

\begin{equation}
	f\left(y\right)<f\left(y_{0}\right)+f'\left(y_{0}\right)\left(y-y_{0}\right)=e^{-y_{0}}y+1-e^{-y_{0}}\left(y_{0}+1\right).
\end{equation}

Also $g\left(y\right)=e^{y}-1$ is an increasing convex function,
such that

\begin{equation}
	g\left(y\right)>g\left(y_{0}\right)+g'\left(y_{0}\right)\left(y-y_{0}\right)=e^{y_{0}}y-1-e^{y_{0}}\left(y_{0}-1\right).
\end{equation}

We can now apply these conditions to Eq. (\ref{eq:SIS_new}) to obtain

\begin{equation}
	\begin{split}\dot{y}_{i}\left(t\right)< &\, \beta e^{-y_{0}}\sum_{j=1}^{n}A_{ij}y_{j}\left(t\right)-\gamma e^{y_{0}}\sum_{j=1}^{n}\delta_{ij}y_{j}\left(t\right)\\
	& +\beta\left[1-e^{-y_{0}}(y_{0}+1)\right]\sum_{j=1}^{n}A_{ij}
	+\gamma \left[1+e^{y_{0}}(y_{0}-1)\right] \sum_{j=1}^{n} \delta_{ij} \coloneqq\hat{y}_{i}\left(t\right),
	\end{split}
\end{equation}
where we have called the upper bound $\hat{y}_{i}$. Using the notation
settled for the initial conditions, this equation is written as

\begin{equation}
	\dot{\hat{y}}_{i}=\beta q\sum_{j=1}^{n}A_{ij}\hat{y}_{j}-\frac{\gamma}{q}\sum_{j=1}^{n}\delta_{ij}\hat{y}_{j}+\beta(p+q\log q)\sum_{j=1}^{n}A_{ij}-\gamma\,\frac{p+\log q}{q},
\end{equation}
or in matrix-vector form as

\begin{equation}
	\dot{\hat{\mathbf{y}}}\left(t\right)=\left(\beta qA-\dfrac{\gamma}{q}I\right)\hat{\mathbf{y}}\left(t\right)+\left[\beta\left(p+q\log q\right)A-\dfrac{\gamma}{q}\left(p+\log q\right)I\right]\mathbf{u}.\label{eq:SIS_new_model}
\end{equation}

It is straightforward to realize that (\ref{eq:SIS_new_model}) has
the form:

\begin{equation}
	\dot{\hat{\mathbf{y}}}\left(t\right)=B\hat{\mathbf{y}}\left(t\right)+\mathbf{b},
\end{equation}
which can then be solved using the method of variation of parameters.
Therefore we have our main result.
\begin{theorem}
	Let $\mathbf{x}\left(t\right),$ $\tilde{\mathbf{x}}\left(t\right)$
	and \textup{$\mathbf{\hat{x}}\left(t\right)$} be, respectively, the
	solution of the exact, linearized $\dot{\tilde{\mathbf{x}}}\left(t\right)=\left(\beta A-\gamma I\right)\tilde{\mathbf{x}}\left(t\right)$,
	and approximate (\ref{eq:SIS_new_model}) SIS model, with the same
	initial conditions: $\mathbf{x}\left(0\right)=\mathbf{\hat{x}}\left(0\right)=\tilde{\mathbf{x}}\left(0\right)=\mathbf{x}_{0}=p\mathbf{u}$.
	Then,
	
	\begin{equation}
		\mathbf{x}\left(t\right)\preceq\mathbf{\hat{x}}\left(t\right)\preceq\tilde{\mathbf{x}}\left(t\right).
	\end{equation}
\end{theorem}

We will prove this result by two parts using the following Lemmas.
We should remark that this result indicates that the solution of the
approximate (\ref{eq:SIS_new_model}) SIS model represents an upper
bound to the exact solution, which is always below the diverging solution
of the linearized SIS model. Let us now prove the first part of this
results using the following.

\begin{lemma}
	Let \textup{$\mathbf{y}\left(t\right)$ be the transformed solution
		of the SIS model. Let $\hat{\mathbf{y}}\left(t\right)$ be the solution
		of the (\ref{eq:SIS_new_model}) model, then}
	
	\begin{equation}
		\mathbf{y}\left(t\right)\preceq\hat{\mathbf{y}}\left(t\right)=e^{Bt}\left[B^{-1}\mathbf{b}-\log q \, \mathbf{u}\right]-B^{-1}\mathbf{b},\label{eq:Solution}
	\end{equation}
	where
	
	\begin{equation}
		B=\beta qA-\dfrac{\gamma}{q}I
		\label{B}
	\end{equation}
	and
	
	\begin{equation}
		\mathbf{b}=\left[\left(p+q\log q\right)\beta A-\left(p+\log q\right)\dfrac{\gamma}{q}I\right]\mathbf{u}.
		\label{b}
	\end{equation}
\end{lemma}

\begin{remark}
	Let's make some remarks about solution (\ref{eq:Solution}):
	
	(i) if $t=0$: $\hat{{\bf y}}(0)=-\log q\,{\bf u}$ and $\hat{{\bf x}}(0)=p$; 
	
	(ii) if $\gamma=0$, we have $B=q\beta A$ and ${\bf b}=(p+q\log q)\beta A{\bf u}$
	so that $B^{-1}{\bf b}=\left(\frac{p}{q}+\log q\right){\bf u}$. Solution
	(\ref{eq:Solution}) reduces to
	\begin{equation}
		\hat{{\bf y}}(t)=\frac{p}{q}e^{q\beta At}{\bf u}-\left(\frac{p}{q}+\log q\right){\bf u}
	\end{equation}
	which is equal to the solution for the SI Model by LTE in Eq. (\ref{LTE solution}). Let us observe that for $t\to+\infty$, $\hat{{\bf y}}(t)\to+\infty$ and $\hat{{\bf x}}(t)\to1$.
	
	(iii) if $\beta=0$: $B=-\frac{\gamma}{q}I$ and ${\bf b}=-(p+\log q)\frac{\gamma}{q}{\bf u}$
	so that $B^{-1}{\bf b}=(p+\log q){\bf u}$. Solution (\ref{eq:Solution})
	becomes
	\begin{equation}
		\hat{{\bf y}}(t)=p\left(e^{-\frac{\gamma}{q}t}-1\right){\bf u}-\log q\,{\bf u}
	\end{equation}
	Let us observe that for $t\to+\infty$, $\hat{{\bf y}}(t)\to-p-\log q$
	and $\hat{{\bf x}}(t)\to1-qe^{p}$. This bound doesn't converge to
	${\bf 0}$ as $t\to+\infty$ but to ${\bf x}^{\star}=(1-qe^{p}){\bf u}$.
	Observe that $0<1-qe^{p}<p$, as expected, and that $1-qe^{p}\to0$
	as $p\to0$. For instance, $x^{\star}<0.1$ if $p<0.392$; $x^{\star}<0.01$ if $p<0.135$; $x^{\star}<0.001$ if $p<0.044$.
	
	(iv) if $\beta\neq0$ and $\gamma\neq0$, the exponential term in
	equation (\ref{eq:Solution}) can be written as
	\begin{equation}
		e^{Bt}=e^{(\beta qA-\frac{\gamma}{q}I)t}=e^{(\beta qM\Lambda M^{T}-\frac{\gamma}{q}I)t}=e^{M(\beta q\Lambda-\frac{\gamma}{q}I)M^{T}t}=Me^{(\beta q\Lambda-\frac{\gamma}{q}I)t}M^{T}
	\end{equation}
	where $\Lambda$ is the diagonal matrix of the eigenvalues of $A$
	and $M$ is the orthogonal matrix whose columns are the eigenvectors
	of $A$. As $t$ grows to $+\infty$, the diagonal exponential terms
	$e^{(\beta q\lambda_{i}-\frac{\gamma}{q})t}$ grows to $+\infty$
	if $(\beta q\lambda_{i}-\frac{\gamma}{q})>0$. In particular, if $(\beta q\lambda_{1}-\frac{\gamma}{q})<0$,
	no one of these terms grows to $+\infty$ and the epidemic decays.
	Thus, we can identify a threshold given by the following condition:
	
	\begin{equation}
		\beta_{e}=\dfrac{\beta}{\gamma}<\dfrac{1}{q^{2}\lambda_{1}}=\tau
	\end{equation}
such that $\tau=\frac{1}{q^{2}\lambda_{1}}$ is the threshold of this bound solution. Then, if $\beta_{e}<\tau$ the epidemic decays; if $\beta_{e}>\tau$ the epidemic grows. We should remark that as $\lambda_{1}$ increases (and so does the average degree in the network), condition above become stricter and the spread of epidemics is facilitated. Moreover, in general, $\tau$ is bigger than $1/\lambda_{1}$, which is the threshold in the exact solution of SIS Networked Model, and it approaches such a threshold as $p$ decreases.
	
\end{remark}

\begin{lemma}
	Let $\tilde{\mathbf{x}}\left(t\right)$ be the solution of the linearized
	SIS problem $\dot{\tilde{\mathbf{x}}}\left(t\right)=\left(\beta A-\gamma I\right)\tilde{\mathbf{x}}\left(t\right)$.
	Then
	
	\begin{equation}
		\mathbf{\hat{x}}\left(t\right)\preceq\tilde{\mathbf{x}}\left(t\right).
	\end{equation}
\end{lemma}

\begin{proof}
	We focus on the above-the-threshold behavior, i.e., we assume
	
	\[
	\beta_{e}=\frac{\beta}{\gamma}>\frac{1}{q^{2}\lambda_{1}}=\tau.
	\]

Then, because $q<1$, we have that $\beta_{e}=\frac{\beta}{\gamma}>\frac{1}{\lambda_{1}}$.
Following Lemma A.1 by LTE (see Eq. (29) in \cite{Lee_et_al}), since the initial conditions are
the same for the bound solution and the linearized process, i.e., $\mathbf{\hat{x}}\left(0\right)=\tilde{\mathbf{x}}\left(0\right)=\mathbf{x}_{0}=p\mathbf{u}$,
it is enough to prove that
	
	\begin{equation}
		\dfrac{d\mathbf{\hat{x}}\left(t\right)}{dt}\preceq\dfrac{d\tilde{\mathbf{x}}\left(t\right)}{dt}
	\end{equation}
for all $t\geq0$. Let us remind that $\mathbf{\hat{x}}\left(t\right)=1-e^{-\mathbf{\hat{y}}}$; then
we have
	
	\begin{equation}
		\dfrac{d\mathbf{\hat{x}}\left(t\right)}{dt}=e^{-\mathbf{\hat{y}}}\dfrac{d\mathbf{\hat{y}}\left(t\right)}{dt}\preceq\dfrac{d\hat{\mathbf{y}}\left(t\right)}{dt},
	\end{equation}
for all $t\geq0$, where the inequality follows from $e^{-\hat{y}_{i}}<1$
for all $\hat{y}_{i}\in[0,\infty]$. By (\ref{eq:Solution}) we have
	
	\[
	\dfrac{d\mathbf{\hat{y}}\left(t\right)}{dt}=e^{Bt}B\left[B^{-1}{\bf b}-\log q\,{\bf u}\right]=e^{Bt}\left[{\bf b}-\log q\,B{\bf u}\right],
	\]
where $B$ and ${\bf b}$ are given by Eq. (\ref{B}) and Eq. (\ref{b}), respectively. Since
	
	\[
	{\bf b}-\log q\,B{\bf u}=p\left(\beta A-\frac{\gamma}{q}I\right){\bf u},
	\]
we have
	
	\begin{equation*}
	\dfrac{d\mathbf{\hat{y}}\left(t\right)}{dt}=e^{\left(\beta qA-\frac{\gamma}{q}I\right)t}\left[p\left(\beta A-\frac{\gamma}{q}I\right){\bf u}\right]\preceq e^{\left(\beta A-\gamma I\right)t}\left[p\left(\beta A-\gamma I\right){\bf u}\right]=\dfrac{d\tilde{\mathbf{x}}\left(t\right)}{dt},
	\end{equation*}
where the last inequality is justified by the fact that $e^{\left(\beta q\lambda_{i}-\frac{\gamma}{q}\right)t}<e^{\left(\beta\lambda_{i}-\gamma I\right)t}$
and $\left(\beta A-\frac{\gamma}{q}I\right){\bf u}\preceq\left(\beta A-\gamma I\right){\bf u}$.
This finally proves the result.
\end{proof}

\begin{remark}
	It is well-known that the linearized SIS model approaches the exact
	solution only when $t\rightarrow0$. This is the case particularly
	when $\beta$ and $\gamma$ are both small, and $\beta>\gamma$ as
	can be seen in Fig. \ref{Toy_evolution}(a), which refers to the toy network in Fig. \ref{Toy network} with parameters $\beta=0.004$, $\gamma=0.001$
	and $p=1/12.$ Notice that we have used logarithmic scale in the time
	to specially highlight the short time behavior. In this case, it can
	be seen that the upper bound found here also coincides with the exact
	solution for short times. For longer times, the linearized solution
	quickly diverges, while our upper bound behaves appropriately and
	converges to the steady state almost as the same time as the exact
	solution. When, $\beta<\gamma$ and both values are not so small,
	the situation is pretty different from the previous one for the linearized
	model. In this case, not even for very small times, the linearized
	solution coincides with the exact one as can be seen in Fig. \ref{Toy_evolution}(b)
	for $\beta=0.02$, $\gamma=0.03$ and $p=5/12.$ The upper bound found
	here coincides with the exact solution for a relatively long period
	of time and then converges to a steady state far from the 100\% of
	contagion as expected for these given set of parameters.
\end{remark}

\begin{figure}[!htbp]
	\begin{centering}
		\subfloat[]{\begin{centering}
				\includegraphics[width=0.47\textwidth]{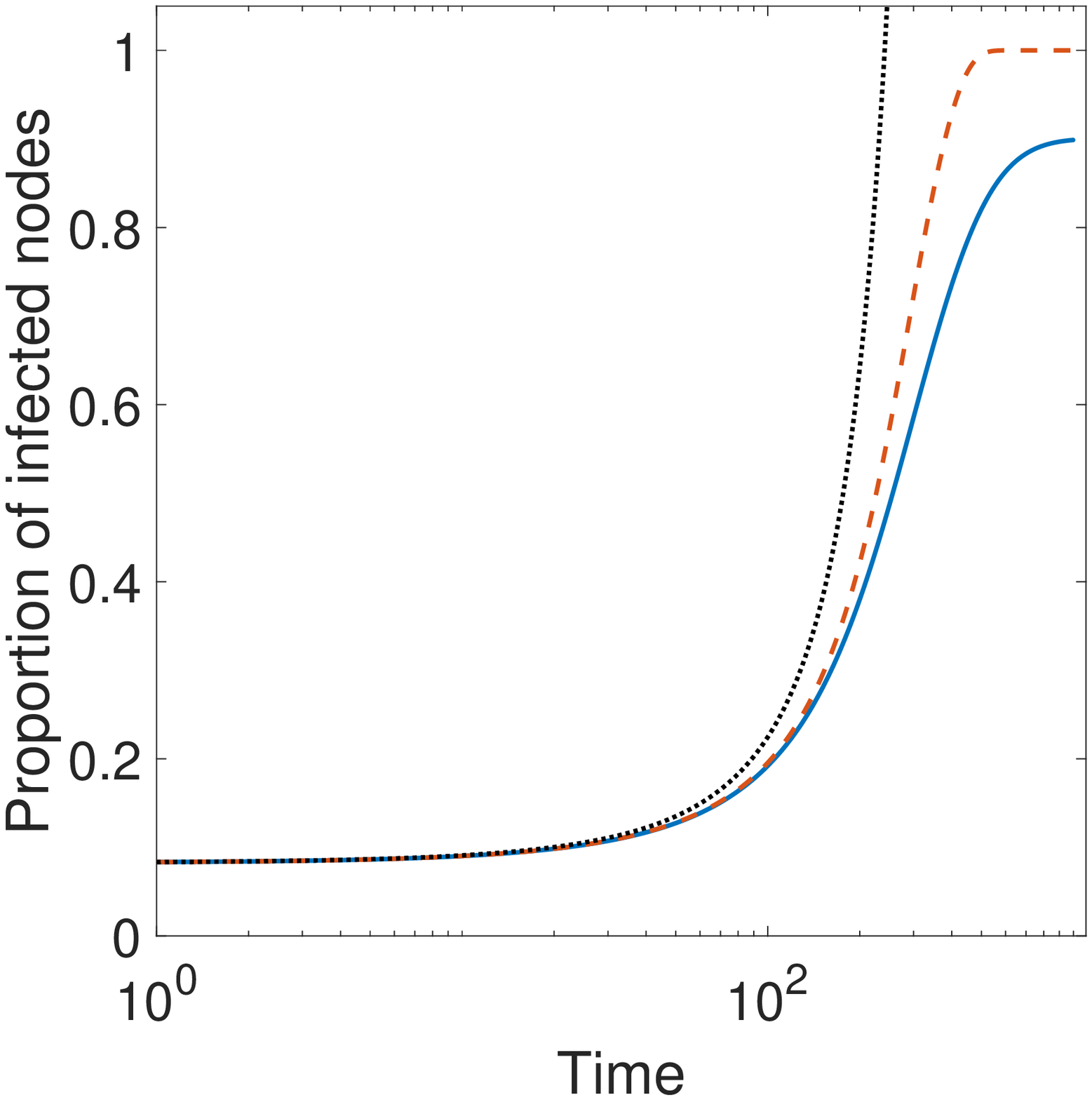}
				\par\end{centering}
		}\subfloat[]{\begin{centering}
				\includegraphics[width=0.50\textwidth]{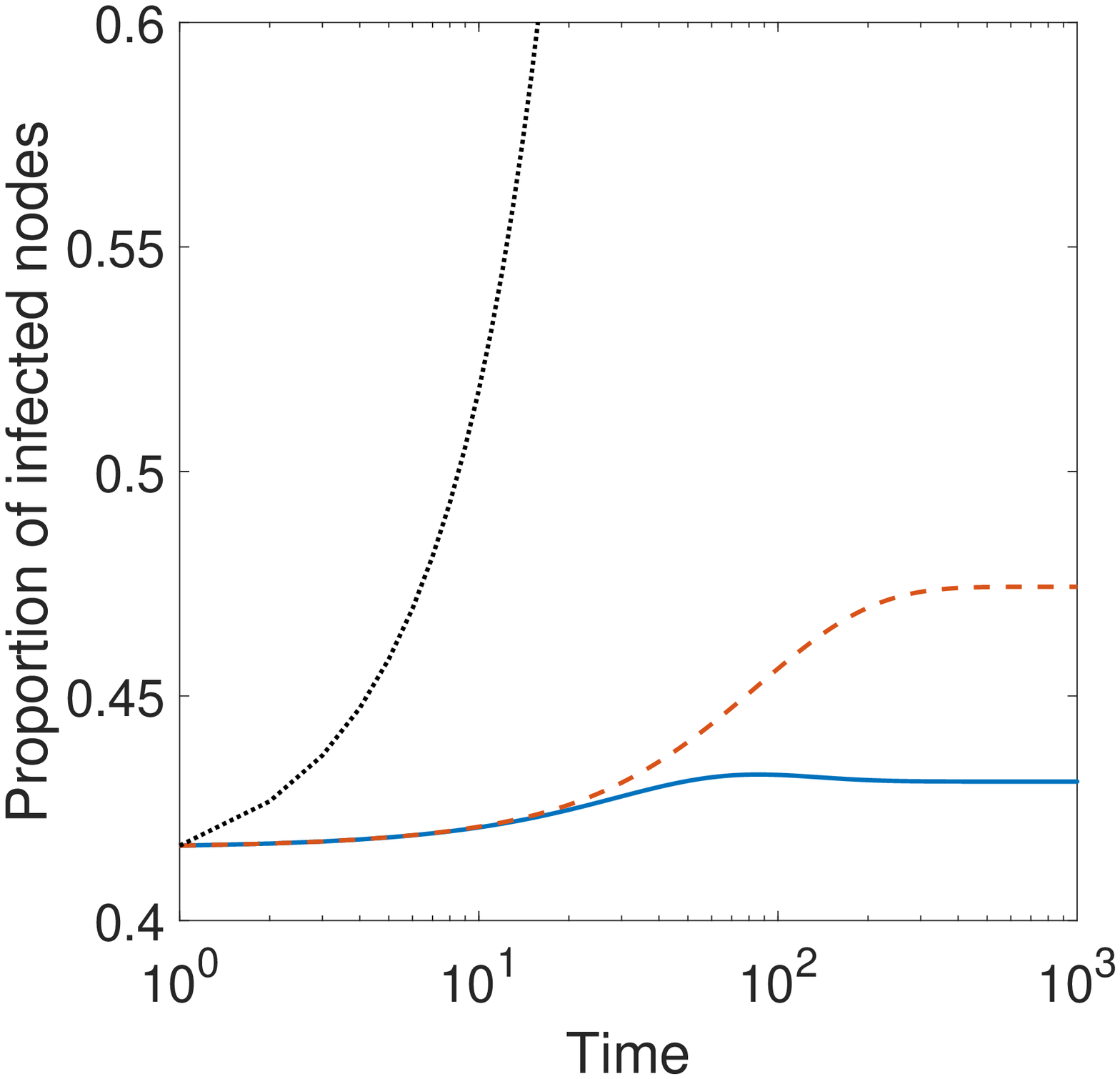}
				\par\end{centering}
		}
		\par\end{centering}
	\caption{Time evolution of proportion of infected nodes in a SIS epidemic on
		the toy network illustrated in Fig. \ref{Toy network}. In blue (continuous)
		line we plot the exact SIS solution, in red (broken) line the solution
		provided by the upper bound found here and in black (dotted) line
		the solution to the linearized model. The time (x-axis) is shown in
		logarithmic scale. In (a) we have $\beta=0.004$, $\gamma=0.001$
		and $p=1/12.$ In (b) $\beta=0.02$, $\gamma=0.03$ and $p=5/12.$}
	
	\label{Toy_evolution}
\end{figure}

\begin{remark}
	We can summarize the strength points of this bound as follows: (i)
	it is an upper bound for any $\beta$, $\delta$ and $p$, (ii) it
	extends to SIS model the upper bound for SI model by LTE, and so it
	generalizes it; (iii) it captures the presence of a threshold $\tau$
	consistent with classic models; (iv) for small initial probabilities
	$p$ it gives a close approximation of the exact solution and a very
	accurate description of the real spreading phenomenon, both above
	or under the threshold. At the same time, it should be taken into
	account that, for big initial probabilities $p$, the solution
	still remains an upper bound, but the approximation gets worse as
	$p\to1$. Moreover, the bound predicts that $\hat{{\bf x}}(t)\to 0$, below the
	threshold, only for small $p$.
\end{remark}

\subsection{Epidemic spread and network communicability}

One of the main goals of network theory is to understand dynamical
processes in terms of structural network properties. Here we use the
upper bound found in this section for the SIS model to understand
how the communicability between nodes in a network captures the propagation
of a disease through the nodes and edges of the network. We start
by setting $B=q\beta A-\frac{\gamma}{q}I=-\frac{\gamma}{q}\left(I-q^{2}\beta_{e}A\right)=:-\frac{\gamma}{q}D$.
Then we can rewrite the vector $B^{-1}{\bf b}-\log q\,{\bf u}$ in
the following way:

\[
B^{-1}{\bf b}-\log q\,{\bf u}=\frac{p}{q}\left[I-p(I-q^{2}\beta_{e}A)^{-1}\right]{\bf u}=\frac{p}{q}\left[I-pD^{-1}\right]{\bf u},
\]
where $D=I-q^{2}\beta_{e}A.$ In this way, solution (\ref{eq:Solution})
becomes

\begin{equation}
	\hat{\mathbf{y}}\left(t\right)=\dfrac{p}{q}\left[e^{-\tfrac{\gamma}{q}Dt}-I\right]\left[I-pD^{-1}\right]{\bf u}-\log q \, {\bf u}.
\end{equation}

Let $\mathbf{v}=\left[I-pD^{-1}\right]{\bf u}$ with elements $v_{i}=\sum_{j=1}^{n}\left[I-pD^{-1}\right]_{ij}$
and let us set

\begin{equation}
	C=\left\{ \begin{array}{cc}
		\max_{i}v_{i} & \textnormal{if}\ \beta_{e}>\tau,\\
		\min_{i}v_{i} & \textnormal{if}\ \beta_{e}<\tau.
	\end{array}\right.
\end{equation}

We then have the following result.
\begin{lemma}
	The probability that a node $i\in V$ in a graph is infected at a
	given time $t$ is bounded by its total communicability $\mathscr{R}_{i}$
	\cite{Benzi2013} as
	
	\begin{equation}
		x_{i}\left(t\right)\leq1-q\exp\left(-C\dfrac{p}{q}\dfrac{\mathscr{R}_{i}\left(\varrho\right)-e^{\tfrac{\gamma}{q}t}}{e^{\tfrac{\gamma}{q}t}}\right),
	\end{equation}
	where $\mathscr{R}_{i}\left(\varrho\right)=\left(e^{\varrho A}{\bf u}\right)_{i}$
	and $\varrho=q\beta t$.
\end{lemma}

\begin{proof}
	Based on the previous definitions we have that
	\begin{equation}
		\mathbf{y}\left(t\right)\preceq\hat{\mathbf{y}}\left(t\right)\preceq C\frac{p}{q}\left[e^{-\tfrac{\gamma}{q}Dt}-I\right]{\bf u}-\log q \, {\bf u}\eqqcolon\bar{\mathbf{y}}\left(t\right).
	\end{equation}

For a given node $i\in V$ we have
	
	\begin{equation}
		\begin{split}{\bar y}_{i}\left(t\right) & =C\frac{p}{q}\left[\left(e^{-\frac{\gamma}{q}Dt}{\bf u}\right)_{i}-1\right]-\log q\\
			& =C\frac{p}{q}\dfrac{\left(e^{q\beta At}{\bf u}\right)_{i}-e^{\frac{\gamma}{q}t}}{e^{\frac{\gamma}{q}t}}-\log q\\
			& =C\frac{p}{q}\dfrac{\mathscr{R}_{i}\left(\varrho\right)-e^{\frac{\gamma}{q}t}}{e^{\frac{\gamma}{q}t}}-\log q,
		\end{split}
	\end{equation}
	from which the solution immediately follows.
\end{proof}

We should notice that, if $\gamma\to0$ then $C\to 1$ and ${\bar x}_{i}(t)\to1-qe^{-\frac{p}{q}({\mathscr{R}}_{i}-1)}$,
which equals the LTE \cite{Lee_et_al} solution of the SI model. We
should also remark that, in this approximation, all the structural information
determining the dynamics of the SIS process is stored in the variable
$\mathscr{R}_{i}$. We can interpret structurally this index as follow.
Let $G\left(\varrho\right)=\exp\left(\varrho A\right)$. Then, $\mathscr{R}_{i}\left(\varrho\right)=\sum_{j=1}^{n}G_{ij}\left(\varrho\right),$
where $G_{ij}\left(\varrho\right)=\left(\exp\left(\varrho A\right)\right)_{ij}=\sum_{k=0}^{\infty}\dfrac{\left[\left(\varrho A\right)^{k}\right]_{ij}}{k!}$
\cite{MF_4,MF_2,MF_3,MF_1}. The important thing here is that $\left(A^{k}\right)_{ij}$
counts the number of walks of length $k$ between the nodes $i$ and
$j$. A walk of length $k$ is any sequence of (not necessarily different)
vertices $v_{1},v_{2},\ldots,v_{k},v_{k+1}$ such that for each $i=1,2,\ldots,k$
there is an edge from $v_{i}$ to $v_{i+1}$ \cite{Estrada2012book}.
When $i=j$ the walk is known as closed. In the expression of ${\bar y}_{i}\left(t\right)$,
$\mathscr{R}_{i}\left(\varrho\right)$ describes every trajectory
of the infective particle starting at the node $i$ (and ending elsewhere)
at a given time and under the fixed initial conditions. It is clear
that with all the epidemiological parameters fixed, ${\bar y}_{i}\left(t\right)$
depends linearly on $\mathscr{R}_{i}\left(\varrho\right).$

\subsection{Network capacity to reroute goods/items/passengers}

We consider that on the same graph $\varGamma$ where the infective
particle is diffusing, there are other desirable diffusive processes
taking place, such as the diffusion of goods, items, and passengers, which move
the economy. We consider that such goods are diffusing through the
graph by means of all available walks connecting a pair of nodes $i$
and $j$ in exactly the same way as the infective particle is using
them. Indeed, there are approaches to modeling the traffic on networks
which use epidemiological models like SIS \cite{traffic_2,traffic_1,traffic_3}.

Therefore, let us find which routes are more probable for the infective particle to travel through. They will be the same ones for the goods/items/passengers moving in the network. For that we start by defining the following difference:
\begin{equation}
	\xi_{ij}\left(\varrho\right)\coloneqq G_{ii}\left(\varrho\right)+G_{jj}\left(\varrho\right)-2G_{ij}\left(\varrho\right),
\end{equation}
which represents the difference between all those walks that start
and end at the same vertex to those walks which go from one node to
another \cite{Benzi2013}. The first two terms then represent the
circulability of a diffusive particle around a given node, while the
last represents the transmissibility between two nodes \cite{bartesaghi2020risk}.
Before continuing we need a clarification here. In the definition
of $\xi_{ij}\left(\varrho\right)$ we are considering the parameter $\varrho$
used in SIS model of the disease propagation. However, we are expecting
that this parameter $\xi_{ij}\left(\varrho\right)$ captures the mobility
of goods/items/passengers in the network, not of the disease. In the particular
cases we are studying here we do not have an estimation for the parameter
$\varrho$ for the mobility of goods/items/passengers. Additionally, we have
the problem that the measure $\xi_{ij}\left(\varrho\right)$ is dependent
on $\varrho$. Therefore, to make comparable the results of the viral
spreading and the mobility of goods/items/passengers we use in the calculations
of the last the same parameter $\varrho$ as for the first. The theoretical
justification for this assumption is that we need $\varrho\ll1$ for
the mobility of goods/items/passengers to avoid congestion problems at the nodes
and the values used for the SIS dynamics fulfill this requirement.
In \cite{estrada2012communicability} it was proved the following
result.

\begin{lemma}
	Let $\xi_{ij}\left(\varrho\right)$ for $\varrho\in\mathbb{R}$ be
	the difference between the circulabilities of a diffusive particle
	around the nodes $i$ and $j$, and the transmissibility between both
	nodes. Then, $\xi_{ij}\left(\varrho\right)$ is a Euclidean distance
	between the corresponding nodes.
\end{lemma}

We should recall that both terms, circulability and transmissibility,
contribute positively to the infection propagation through: $\mathscr{R}_{i}\left(\varrho\right)=G_{ii}\left(\varrho\right)+G_{ij}\left(\varrho\right)+\sum_{k\neq i\neq j}^{n}G_{ik}\left(\varrho\right)$.
Therefore, we aim here at the following task:
\begin{itemize}
	\item How to decrease significantly $\mathscr{R}_{i}\left(\varrho\right)$,
	and consequently ${\bar x}_{i}\left(t\right)$, without increasing
	significantly $\xi_{ij}\left(\varrho\right)$, and consequently minimally
	affecting the network capacity to diffuse goods/items/passengers?
\end{itemize}
Although $\xi_{ij}\left(\varrho\right)$ could be a good proxy for
the network capacity of transporting goods/items/passengers we should be aware
that all the transport taking place in a network occurs through the paths
connecting two nodes. That is, $\xi_{ij}\left(\varrho\right)$ does
not necessarily indicate the route followed by an item from the node $i$ to the
node $j$. For finding such routes we need a geometrization of the
graph. This is carried out by defining a length space on it \cite{Geometrization_2,Geometrization_1}.
Let us consider $e=(i,j)$ as a compact 1-dimensional manifold with boundary
$\partial e=i\cup j$. Let the edge $e=(i,j)$ be given the $\xi_{ij}\left(\varrho\right)$
metric such that
\begin{equation}
	e_{ij}\underset{isom}{\cong}\left\{ \begin{array}{cc}
		\left[0,\xi_{ij}\left(\varrho\right)\right] &\ \left(i,j\right)\in E\\
		0 &\ \ \left(i,j\right)\notin E.
	\end{array}\right.
\end{equation}

We now extend the metric on the edges of $\varGamma$ via infima of
lengths of curves in the geometrization of $\varGamma$. Then, the
network becomes a metrically length space, which is locally compact,
complete and geodetic \cite{Geometrization_2}. Now define the ``shortest
diffusive path length'' as:
\begin{equation}
	\mathscr{C}_{ij}\left(\varGamma, \varrho\right)\coloneqq\min_{P,ij}\sum_{\substack{\left(i,j\right)=e\in E\\
			e\in P
		}
	}\xi_{ij}\left(\varrho\right),
\end{equation}
where $P$ is a path in $\varGamma$, i.e., a walk with repetition
neither of vertices nor of edges, and the minimum is taken among all
paths connecting the corresponding pairs of vertices. We can now define
the capacity of a network to reroute goods/items/passengers after the removal
of an edge $e$ by:
\begin{equation}
	\varDelta\bar{\mathscr{C}}\left(\varGamma -e, \varrho\right)=\dfrac{\bar{\mathscr{C}}\left(\varGamma-e, \varrho\right)-\bar{\mathscr{C}}\left(\varGamma, \varrho\right)}{\bar{\mathscr{C}}\left(\varGamma, \varrho\right)},
\end{equation}
where $\mathcal{\mathscr{\bar{C}}}\left(\varGamma, \varrho\right)$ is the mean shortest communicability path (SCP), that is the average of $\mathscr{C}_{ij}\left(\varGamma, \varrho\right)$ over all shortest diffusive
paths connecting pairs of nodes in $\varGamma$, and $\varGamma-e$ is the graph from which the edge $e$ has been removed.

\subsection{Implementation of an edge-removal strategy}

We are always interested in nontrivial edge removals here, i.e., those
that do not disconnect the graph. Then, to respond to the main query
formulated in the previous subsection we will consider edge-removal
strategies that decrease significantly $\mathscr{R}_{i}\left(\varrho\right)$,
but do not increase significantly $\varDelta\bar{\mathscr{C}}\left(\varGamma -e,\varrho\right).$
It is obvious that any strategy that increases the relative communicability
between two vertices $G_{ij}\left(\varrho\right)$ will necessarily
drops $\xi_{ij}\left(\varrho\right).$ Unfortunately, it will also
increase $\mathscr{R}_{i}\left(\varrho\right).$ The obvious strategy
seems to drop $G_{ii}\left(\varrho\right)$ so that both $\mathscr{R}_{i}\left(\varrho\right)$
and $\varDelta\bar{\mathscr{C}}\left(\varGamma -e,\varrho\right)$ diminish their values.
However, very frequently dropping $G_{ii}\left(\varrho\right)$ also
decreases $G_{ij}\left(\varrho\right).$ Therefore, we cannot foresee
at first hand a strategy that fulfill both requirements and we then
implement a computational approach to investigate the problem.
First, we start by defining the following term:
\begin{equation}
	\mathcal{O}\left(e,t\right)\coloneqq\max_{e\in E}\dfrac{\varDelta t^{\star}}{\varDelta\bar{\mathscr{C}}\left(\varGamma -e,\varrho\right)},
\end{equation}
where $\varDelta t^{\star}=\left(t^{\star}\left(\varGamma-e\right)-t^{\star}\left(\varGamma\right)\right)/t^{\star}\left(\varGamma\right)$
and $t^{\star}$ is the time at which every node in the network is
infected, i.e., the steady state of the SIS process, where the maximum
is obtained among all the edges of the graph.

We must be aware of an important characteristic of this process. The
term $\varDelta\bar{\mathscr{C}}\left(\varGamma -e,\varrho\right)$ depends on $t$, which
means that $\mathcal{O}\left(e,t\right)$ is different for different
times. This means that the process of edge-removal is time-dependent,
and we should go removing edges as the time of the evolution of the
epidemic goes on. This is a very realistic scenario and reflect some
of the difficulties found in the current COVID-19 pandemics, where
the measures taken at a given time are not necessarily the optimal
ones at another. In Algorithm \ref{algorithm} we provide the pseudo-code of the current
implementation.

\begin{algorithm}
	\DontPrintSemicolon
	\KwIn{The original network: $\Gamma=(V,E)$, $v_i\in V$ and $(i,j)\in E$;\\ $\ \quad \qquad$ The downdating times: $T_{k}=k\cdot a,\ k\in [0,m]$, $a$ fixed time step;\\ $\ \quad \qquad$ The epidemic level $\varepsilon$.} 
	\KwOut{The downdated network after $m$ steps: $\Gamma_{m}=(V,E_{m})$}
	set $\Gamma_{0}=\Gamma$, $A_{0}=A$, $X_{0}(t)=X(t)$ and $t_{0}=\min \left( t: X_{0}(t)\geq 1-\varepsilon \right)$\;
	\For{$k \in [1,m]$}{
		$A_{k}^{(i,j)} \gets A_{k-1}-A_{k-1}\cdot U_{ij}U_{ji}^{T}, \ \forall i,j \in E_{k-1}$\;
		generate $\Gamma_{k}^{(ij)}$ with adjacency matrix $A_{k}^{(ij)}$, $\forall i,j \in E_{k-1}$\;
		if ${\rm count.components}\,\big( \Gamma_{k}^{(ij)}\big)\neq 1$: $A_{k} \gets A_{k-1}$ and stop; else\, (NULL)\;
		$X_{k}^{(ij)}(t) \gets X_{k-1}(t)$, $\forall i,j \in E_{k-1}$\;
		for $t\geq T_{k}$ compute $X_{k}^{(ij)}(t)$ on network  $\Gamma_{k}^{(ij)}$ with $X_{k}^{(ij)}(T_{k})=X_{k-1}(T_{k})$, $\forall i,j \in E_{k-1}$\;
		compute $t_{k}^{(ij)}=\min \left( t: X_{k}^{(ij)}(t)\geq 1-\varepsilon \right), \forall i,j \in E_{k-1}$\;
		compute $\varDelta t_{k}^{(ij)}=t_{k}^{(ij)}-t_{k-1}$\;
		compute $\varDelta {\bar {\mathscr{C}}}_{k}^{(ij)}=\frac{{\bar {\mathscr{C}}}^{(ij)}_{k}(a)-{\bar {\mathscr{C}}}_{k-1}(a)}{{\bar {\mathscr{C}}}_{k-1}(a)}$ where ${\bar {\mathscr{C}}}^{(ij)}_{k}(a)$ is the mean SCP on network $\Gamma_{k}^{(ij)}$ at time $T_{k}=ka$ and ${\bar {\mathscr{C}}}_{k-1}(a)$ is the mean SCP on network $\Gamma_{k-1}$\;
		select $({i_k},{j_k})\in E_{k-1}$ corresponding to $\max_{{i},{j}}\frac{\varDelta t_{k}^{(ij)}}{\varDelta {\bar C}_{k}^{(ij)}}$\;
		remove edge $({i_k},{j_k})\in E_{k-1}$ from $\Gamma_{k-1}$ and generate $\Gamma_{k}$\;
	}
	\Return{$\Gamma_{m}$ and $X_{m}(t)$}\;
	\caption{{\sc Optimal Downdating}}
	\label{algorithm}
\end{algorithm}

\subsubsection{Toy network example}

It is time now to give some numbers and we will start by analyzing
the toy model illustrated in Fig. \ref{Toy network}. The process
evolves as follow. We consider the time evolution of the SIS model
in which we observe the evolution of the ratio of infected nodes with time. At a given time, we make the following plot. For every potential edge-removal
of interest, here made for every of the 16 edges of the graph, we
plot $\varDelta\bar{\mathscr{C}}\left(\varGamma -e,\varrho\right)$ vs. $\varDelta t^{\star}$ in a box as the one illustrated in Fig. \ref{Toy_plots}(a)
for $t=150.$ The points in the plot correspond to the effects produced
by removing the corresponding edge. The radii and color of these points
are proportional to the values of $\mathcal{O}\left(e,t=150\right)$. It
can be clearly seen that there are two groups of edges. In the upper-right
corner we have all the edges whose removal change very much the capacity
of the network to reroute goods/items/passengers. In the opposite corner we have
all those edges whose removal increase the time for infecting the
whole population with minimum disruption of network operational capacity.
We can select here a given number of edges to be removed in dependence
of other factors, of economic or logistical nature. We remove here
one edge at a time. In this case we select the edge with the largest
$\mathcal{O}\left(e,t=150\right)$ which is the edge $\left(2,3\right).$
This single removal increases the time at which the SIS dynamics infects the $90\%$ of
the whole population by 17.1\% with a minimum change in the network
capacity to operate, i.e., $\varDelta\bar{{\mathscr{C}}}\left(\varGamma -e, \varrho\right)\approx0.72\%$.
This edge removal produces a change in the trajectory of the infection
as observed in the plot Fig. \ref{Toy_plots}(d), which is marked
by the point $(2,3)$, which represents the edge removed.

We now continue observing the evolution of the epidemic until we decide
the next intervention. In this case we decided to do it at $t=300$
(we simply use similar periods of time here to make the interventions,
in a real-life situation this can be done at irregular intervals).
Notice that the value of $\varDelta\bar{{\mathscr{C}}}\left(\varGamma -e,\varrho\right)$ is
dependent on the time at which we decide to make the plot. The new
situation is observed in the plot Fig. \ref{Toy_plots}(b) where the
model informs us that the next best cut should be made at edge $\left(2,4\right)$.
The combined interventions of cutting edges $(2,3)$ and $(2,4)$
increases the time to reach 90\% of infected population by 53.8\%.
If we translate this into days, for instance, it means to gain almost
54 days out of 100, which is very significant. Now the capacity of
the network to operate has drop by 3.0\%. The trajectory of the infection
changes again at the point $(2,4)$ of the plot in Fig. \ref{Toy_plots}(d).

In Fig. \ref{Toy_plots}(c) we illustrate the result of the third intervention
at $t=450$, which indicates that the next cut should be made to the
edge $\left(3,4\right).$ The combined interventions which have removed
three edges out of 16 in this toy network increases the time for infecting the $90\%$ of
the whole network by 160.4\%(!). That is, by removing only 18.75\%
of edges, which does not disconnect the graph, we have more than duplicated
the time that we now have to take actions during the epidemic, from 381
time units to 992 ones. All this by dropping only in 5.34\% the total
capacity of the network to operate in relation to normal conditions.
The epidemic now follows the trajectory from the point $\left(3,4\right)$
in the plot in Fig. \ref{Toy_plots}(d).

\begin{figure}[H]
	\subfloat[]{\begin{centering}
			\includegraphics[width=0.5\textwidth]{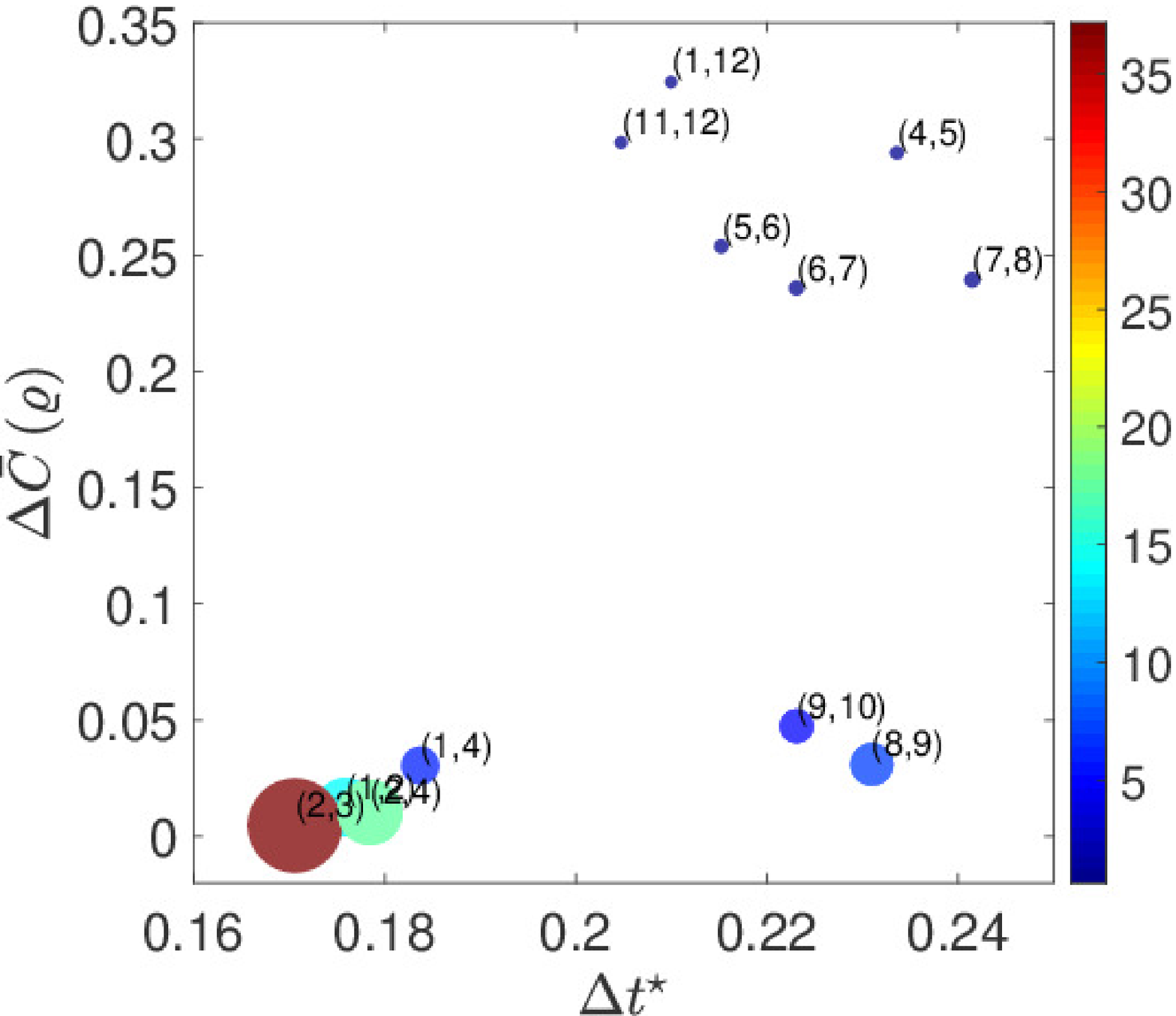}
			\par\end{centering}
		
	}\subfloat[]{\begin{centering}
			\includegraphics[width=0.48\textwidth]{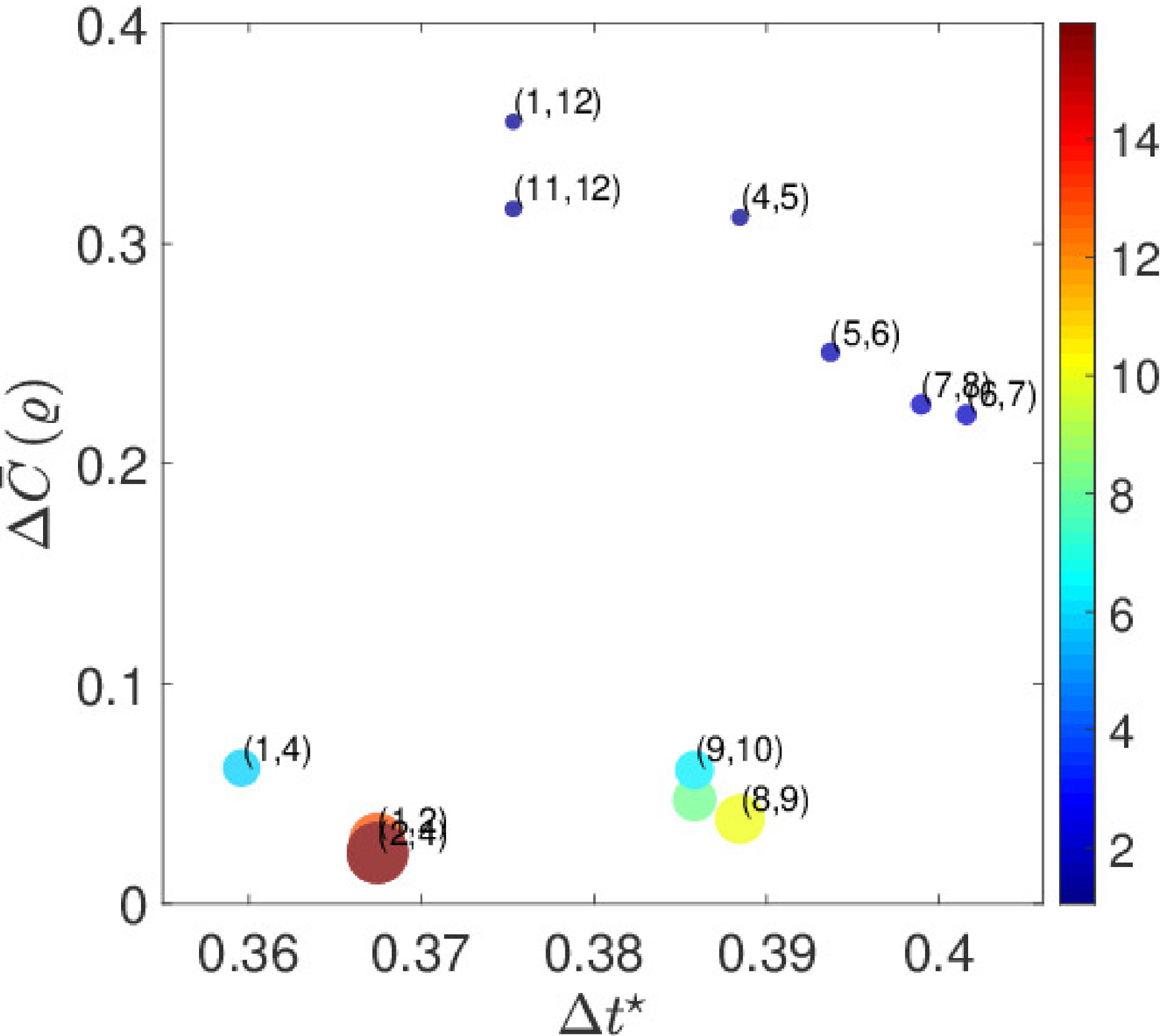}
			\par\end{centering}
	}
	
	\subfloat[]{\begin{centering}
			\includegraphics[width=0.5\textwidth]{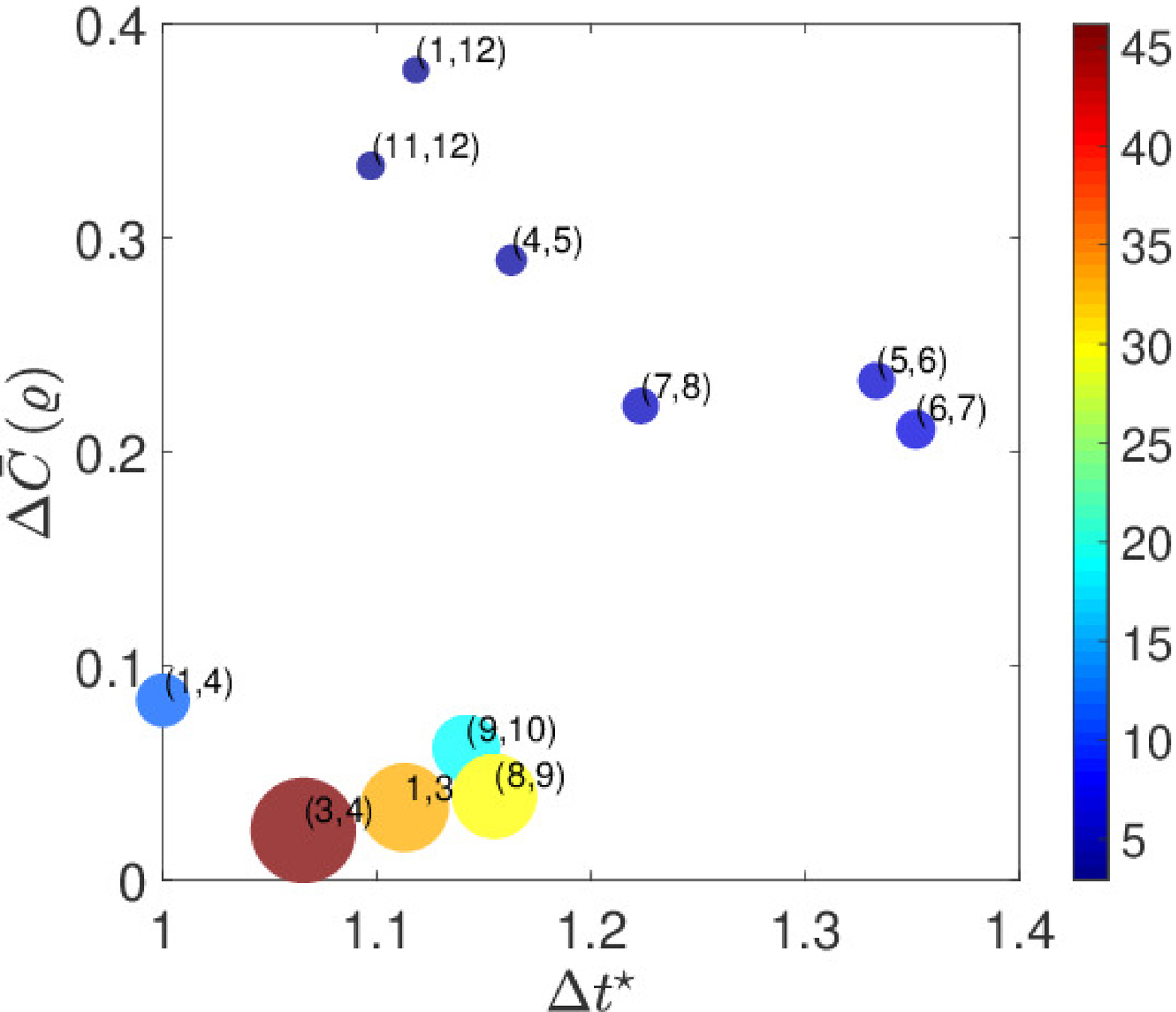}
			\par\end{centering}
	}\subfloat[]{\begin{centering}
			\includegraphics[width=0.45\textwidth]{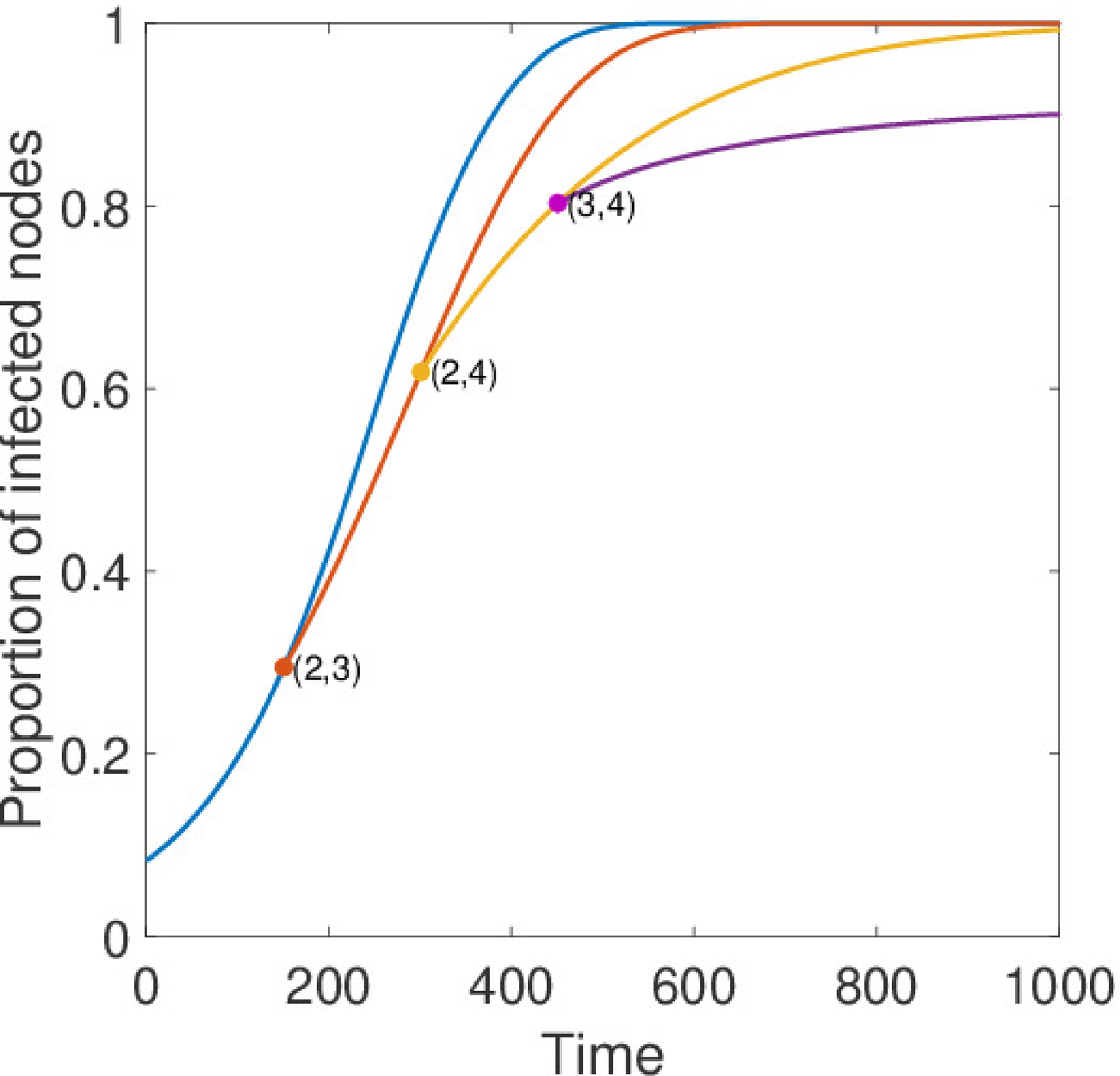}
			\par\end{centering}
	}
	\caption{Plots $\varDelta\bar{\mathscr{C}}\left(\varGamma -e,\varrho\right)$ vs. $\varDelta t^{\star}$: the radii and color of the points are proportional to the values of $\mathcal{O}(e,t)$ at different time units, specifically at a) $t=150$, b) $t=300$, c) $t=450$; figure d) shows the trajectory of the infection curve as a function of time: each marked point refers to the removal of the corresponding edge.}
	\label{Toy_plots}
\end{figure}

\section{Analysis of a real-life situation: the UK airports network}

Here we consider the network of domestic flights between 44 commercial
airports in the United Kingdom in the year 2003. This was the year
in which the SARS epidemic was spreading across the world. The network
consists of 220 weighted edges representing air internal routes between
these 44 airports. The weights correspond to the number of passengers
transported during that year between the corresponding airports. In
Fig. \ref{UK_network} we show a representation of this network where
the nodes are drawn with size and colors proportional to their
weighted degrees--total number of passengers arriving/departing to/from
that airport in 2003. The weighted degree $w_{i}$ and the ``standard''
degree $k_{i}$, i.e., number of edges incident to the node $i$,
are related to each other by means of a power-law relation: $w\approx e^{10.46}k^{1.419}$
with Pearson correlation coefficient $r=0.631$. More importantly,
rank correlation indicates that both indices rank the airports in
very different ways. For instance, the Kendall $\tau$ coefficient
between both indices is only $\tau\approx0.626$. According to $w$
the most ``important'' airport was Heathrow, which was visited this
year by 6,176,092 passengers, representing 13.8\% of all passengers
traveling this year across the U.K. However, the degree of this node
is only 9, possibly indicating its major role as an international
hub and connecting only to those relevant national airports from which
passengers can easily move to other places. On the other side of the
coin we have the airports of Jersey and Aberdeen with degree 24
and 23, respectively. Each of them moved less than 3\% of the total
number of passengers in the U.K. this year. However, Jersey is a major
touristic destination in the U.K. and Aberdeen has become an important
work hub in the U.K. due to the oil industry. Thus, they receive flights
from many different cities across the U.K., although the number of
passengers is relatively low.

\begin{figure}[!htbp]
	\begin{centering}
		\includegraphics[width=0.50\textwidth]{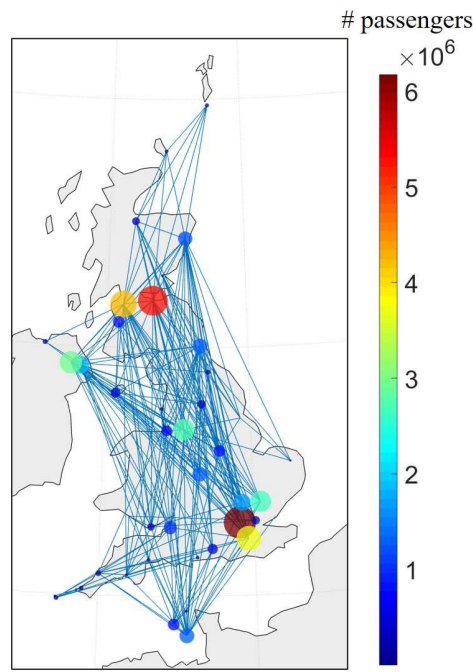}
		\par\end{centering}
	\caption{Network of air transport of passengers in the U.K. The size and color
		of the nodes is proportional to the number of passengers transported
		from/to that airport. Two airports are connected if there is at least
		one flight between the two.}
	
	\label{UK_network}
\end{figure}

Due to these differences between the weighted and unweighted degrees, we
consider here the analysis of both versions of the U.K. airport transportation
network for the propagation of a SIS disease and the implementation
of edge-removal strategies. We apply the edge-removal strategy described
in this work for removals at $t=200,400,600,800,1000$. In the case
of the unweighted network, these removals correspond to the connections
between the following pairs of airports (in order): Glasgow-Manchester;
Belfast City-Manchester; Stansted-Edinburgh; Isle of Man-Manchester
and Bristol-Guernsey. In total, the removal of these 5 connections
increases the cumulative time for infecting the whole network - precisely to overcome the epidemic level of 90\% of infected nodes - by 161.9\% with a decrease of 0.25\% in the capacity of the network to reroute goods/items/passengers.

\begin{table}[H]
	
	\begin{centering}
		\begin{tabular}{|c|c|c|c|c|}
			\hline 
			time & Route & $\varDelta t^{\star}$ (\%) & $\varDelta \bar{{\mathscr C}}$ (\%) & $\varDelta {\bar{l}}$  (\%)\tabularnewline
			\hline 
			\hline 
			200 & Glasgow-Manchester & 9.52 & 0.026 & 0.048\tabularnewline
			\hline 
			400 & Belfast C.-Manchester & 23.81 & 0.068 & 0.101\tabularnewline
			\hline 
			600 & Stansted-Edinburgh & 47.62 & 0.085 & 0.152\tabularnewline
			\hline 
			800 & Isle of Man-Manchester & 89.95 & 0.189 & 0.203\tabularnewline
			\hline 
			1000 & Bristol-Guernsey & 161.90 & 0.247 & 0.254\tabularnewline
			\hline 
		\end{tabular}\caption{Quantitative results of the edge-removal strategy using the unweighted
			version of the U.K. air transportation network.}
		\par\end{centering}
\end{table}

In contrast, the consideration of the number of passengers between
the different airports produces a completely different picture. First,
we need a normalization of the weighted adjacency matrix to make the
edge weights comparable to those of the unweighted version. This is
carried out by dividing the weighted adjacency matrix with the mean
value of the edge weights in the network. In this way, both the unweighted
and the normalized weighted adjacency matrices have the same mean.
Using this strategy the order of removals is as follows: Heathrow-Edinburgh;
Heathrow-Manchester; Heathrow-Glasgow and Heathrow-Belfast City. The
fifth removal is not carried out as it is not necessary to drop to
probability of infecting the whole network below 90\%, which was the
target of the experiment. The evolution of the mean probability that an
airport gets infected at three different times is illustrated in Fig.
\ref{airports_probability_evolution}.

\begin{figure}[H]
	\begin{centering}
		\includegraphics[width=0.9\textwidth]{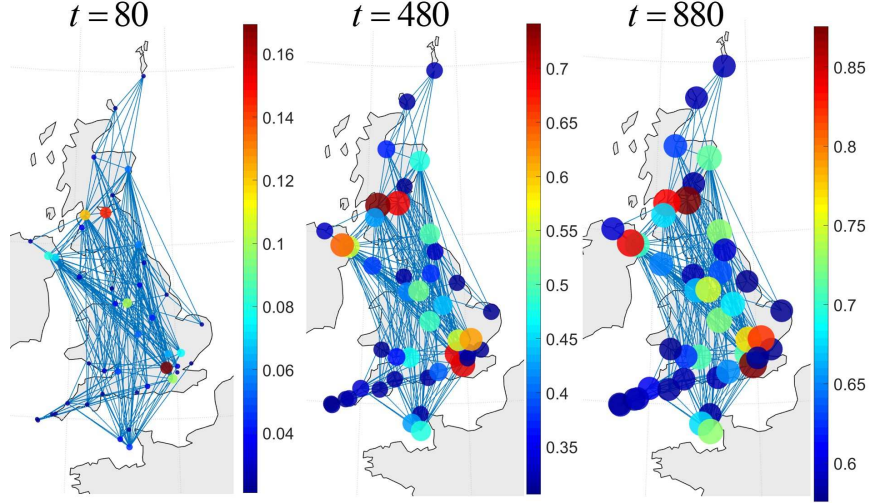}
		\par\end{centering}
	\caption{Illustration of the evolution of the probability of getting infected
		in the U.K. airports network at an initial (no edge removal), intermediate
		(after removals of Heathrow-Edinburgh and Heathrow-Manchester) and
		advanced (after all four edge removals) times. The radii and colors
		of the nodes are proportional to the probability of getting infected
		which is illustrated in colorbars.}
	\label{airports_probability_evolution}
\end{figure}

By removing all the flights between Heathrow and Edinburgh we delay
the time for infecting the whole network by 38.9\%. The second removal
increases this time to 88.9\% and the third one increases it up to
194.4\%. Finally, the cumulative removal of 4 connections increases
the time to infect the whole network by 333.3\%. How the capacity
of the airport network to reroute goods/items/passengers has changed
after these removals? The response is surprising! The remove of all
flights between Heathrow and Edinburgh does not drop the capacity
of the global network to diffuse goods/items/passengers and passengers through
its nodes. In contrast, it increases this capacity by 9.9\%. This
is, of course, a consequence of considering that such goods/items/passengers
move in the network in a completely diffusive way.
If we consider that they move using the shortest paths, then we observe
a drop in the capacity of the network equal to $\varDelta {\bar{l}}=0.05\%$, i.e., the increase
in the average shortest path length in the network after the removal.
After the four removals previously described the network has increased
its diffusive capacity by 3.6\% with a drop in its capacity to deliver
goods/items/passengers via shortest paths of 0.4\%. In either way,
the removal of these four inter-airport connections produces a remarkable
delay on the propagation of the SIS disease in comparison with a very
small affection of the network operative capacity.

\begin{table}[!htbp]
	\centering{}
	\begin{tabular}{|c|c|c|c|c|}
		\hline 
		time & Route & $\varDelta t^{\star}$ (\%) & $\varDelta {\bar{\mathscr C}}$ (\%) &  $\varDelta {\bar{l}}$ (\%)\tabularnewline
		\hline 
		\hline 
		200 & Heathrow-Edinburgh & 38.89 & -9.93 & 0.049\tabularnewline
		\hline 
		400 & Heathrow-Manchester & 88.89 & -9.34 & 0.152\tabularnewline
		\hline 
		600 & Heathrow-Glasgow & 194.44 & -5.61 & 0.254\tabularnewline
		\hline 
		800 & Heathrow-Belfast C. & 333.33 & -3.60 & 0.407\tabularnewline
		\hline 
	\end{tabular}\caption{Quantitative results of the edge-removal strategy using the passengers-weighed
		version of the U.K. air transportation network with adjacency matrix
		normalized by the mean number of passengers in the network.}
\end{table}

\begin{figure}[!htbp]
	\begin{centering}
		\subfloat[]{\begin{centering}
				\includegraphics[width=0.45\textwidth]{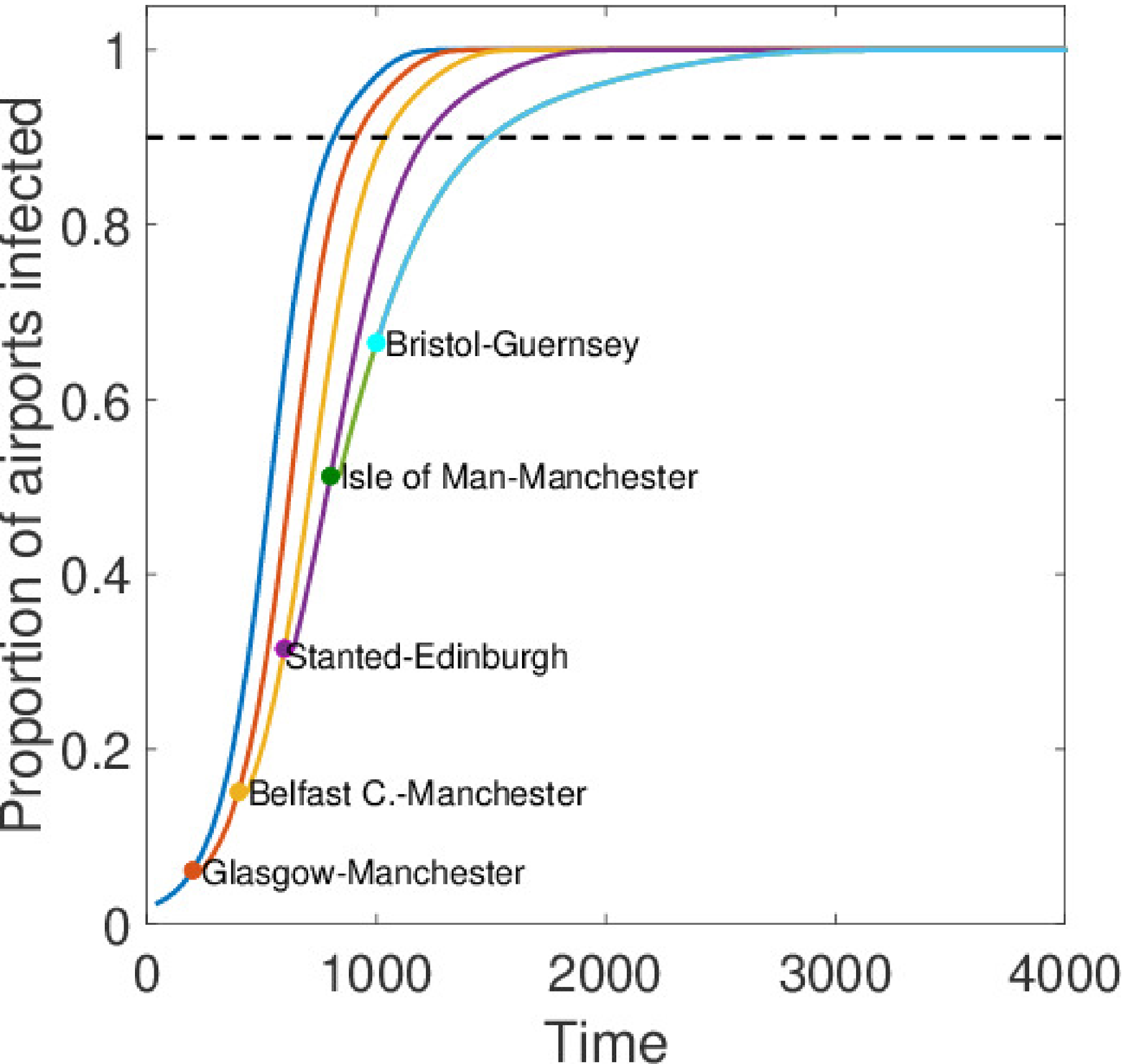}
				\par\end{centering}
			
		}\subfloat[]{\begin{centering}
				\includegraphics[width=0.45\textwidth]{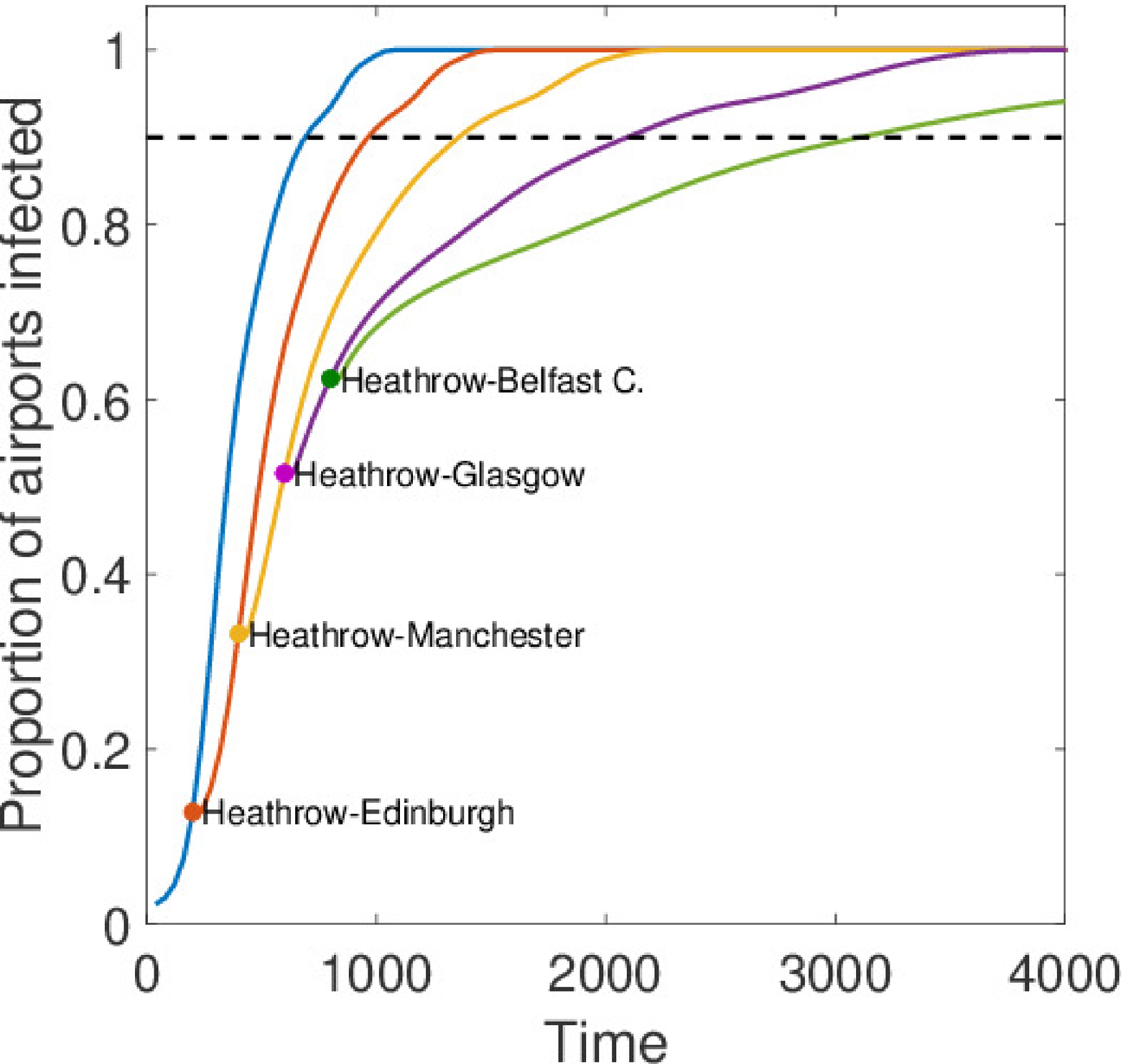}
				\par\end{centering}
		}
		\par\end{centering}
	\caption{Time evolution of the proportion of infected airports subjected to
		edge-removal strategies in (a) the unweighted version of the network,
		and (b) the passengers-weighed network with adjacency matrix normalized
		by the mean number of passengers.}
	
\end{figure}

We have also conducted experiments to show whether the previous results
are dependent on the normalization scheme used for the weighted adjacency
matrix. In this case we use two other normalization schemes, namely
by dividing the weighted adjacency matrix by the maximum edge weight
or by dividing it by the total sum of edge weights. In both cases
the edges identified to be removed are the same as for the case of
normalizing by the mean weight with the addition of a fifth edge
to be removed, which corresponds to Belfast Int.-Liverpool. The time
to infect the whole network increases by 286.4\% and by 309.5\% after
the fifth removal using the two additional schemes of normalization, respectively.
All in all, these experiments show that the results previously described
using the mean-weight normalization of the adjacency matrix are not
specific of this kind of normalization and stressed the importance
of using passengers-weighted version of the airport networks.

\section{Conclusions}

We have developed an approximate solution to the SIS epidemiological
model which represents an upper bound to the exact solution of that
model. This upper bound has several important features: (i) it does
not diverge as the linearized SIS model; (ii) it represents a worse-case
scenario for the propagation of a SIS disease; (iii) its solution
can be expressed in terms of the communicability function, allowing
clear structure-dynamic relations. Using this model and its connection
with the communicability function, we proposed here a general strategy
for mitigating the effects of a disease propagation on a network with
minimum disruption of network capacities to reroute goods/items/passengers.
This strategy consists in removing some connections which are found
to delay the propagation of a disease on the network but minimally
altering the capacity of the network to diffuse items among its nodes
or to reroute them by alternative shortest paths. As a proof of concept,
we have studied the airport transportation network of U.K. in 2003,
where the nodes represent airports and the edges represent the flight
connections, weighed by the number of passengers transported this
year, between them. We have shown that using the strategy proposed
in this work, the removal of only 4 origin-destination pairs in a
time-dependent way delays the propagation of an epidemic by more than
330\% relative to the original network. This delay represent a very
significant gain in time for preparations of health systems and non-pharmaceutical
interventions to confront such epidemic. In addition, these removals
alter minimally the capacity of the U.K. airport system to transport
goods/items/passengers either in diffusive ways or via shortest-paths
routing.

Of course, the global spread of a virus, like SARS-Cov-2, is a multi-scale phenomenon and cannot be faced by any individual oversimplified strategy but it requires action at various different levels.\cite{Cirillo2020, Fatehi2021, Pitcher2018} It is worth mentioning here the wider scopes of our intentions, not only related to epidemiology and virus propagation. The proposed methodology is flexible enough to adapt to very different contexts, like for instance dissemination of information, risk assessment in finance, and beyond.

Indeed, the main emphasis of the current work has been on the mathematical, methodological side. We expect that the extension of this approach to other epidemiological models allow more realistic implementations to tackle this important kind of non-pharmaceutical interventions that mitigate the effects of epidemics in the future.

\section*{Acknowledgments}

The author thanks financial support from Ministerio de Ciencia, Innovacion y Universidades, Spain for the grant PID2019-107603GB-I00 ''Hubs-repelling/ attracting Laplacian operators and related dynamics on graphs/networks''.

\bibliographystyle{ws-m3as}
\bibliography{references}

\end{document}